\documentclass[10pt]{article}
\usepackage{verbatim,color,amssymb,epsfig}
\usepackage{amsmath,amsthm,graphicx,bbm,url}
\usepackage[usenames,dvipsnames]{xcolor}
\usepackage{fancyhdr}
\usepackage[authoryear, sort]{natbib}

\newcommand{\Var}{\text{Var }}
\newcommand{\vc}[0]{\text{vec}}
\newcommand{\argmin}[1]{\underset{#1}{\operatorname{argmin}}\text{ }}

\newcommand{\MISE}{\text{MISE}}
\newcommand{\w}{\alpha}
\newcommand{\page}{p.}
\newcommand{\hsq}{s}
\newcommand{\no}{NO$_2$ }

\newcommand{\fyh}{\widetilde{f}_{Y,H}}
\newcommand{\fyhy}[2]{\widetilde{f}_{Y,#1}(#2)}
\newcommand{\fy}{\fyhy{H}{y}}

\newcommand{\fxhx}[2]{\widetilde{f}_{X,#1}(#2)}

\newcommand{\ffyh}{\, \widetilde{\! \! \widehat{f}}_{Y,H}}
\newcommand{\ffy}{\, \, \widetilde{\! \! \widehat{f}}_{Y,H}(\omega)}

\newcommand{\ffx}{\, \, \widetilde{\! \! \widehat{f}}_{X}(\omega)}

\newtheorem{Thm}{\underline{\bf Theorem}}
\newtheorem{Assume}{\underline{\bf Assumptions}}

\newtheorem{Lem}{\underline{\bf Lemma}}

\def\E{\mathbb{E}}

\begin{document}
\thispagestyle{empty}
\baselineskip=28pt
\begin{center}
{\LARGE{\bf Kernel Density Estimation with Berkson Error}}
\end{center}

\baselineskip=12pt

\vskip 2mm
\begin{center}
James P. Long\\
Department of Statistics, Texas A\&M University\\
3143 TAMU, College Station, TX 77843-3143\\
jlong@stat.tamu.edu\\
\hskip 5mm \\
Noureddine El Karoui \\
Department of Statistics, University of California, Berkeley\\
367 Evans Hall \# 3860, Berkeley, CA 94720-3860\\
nkaroui@stat.berkeley.edu\\
\hskip 5mm \\
John A. Rice \\
Department of Statistics, University of California, Berkeley\\
367 Evans Hall \# 3860, Berkeley, CA 94720-3860\\
rice@stat.berkeley.edu
\end{center}

\hskip 5mm

\begin{center}
{\Large{\bf Abstract}}
\end{center}
\baselineskip=12pt
Given a sample $\{X_i\}_{i=1}^n$ from $f_X$, we construct kernel density estimators for $f_Y$, the convolution of $f_X$ with a known error density $f_{\epsilon}$. This problem is known as density estimation with Berkson error and has applications in epidemiology and astronomy. Little is understood about bandwidth selection for Berkson density estimation. We compare three approaches to selecting the bandwidth both asymptotically, using large sample approximations to the $\MISE$, and at finite samples, using simulations. Our results highlight the relationship between the structure of the error $f_{\epsilon}$ and the optimal bandwidth. In particular, the results demonstrate the importance of smoothing when the error term $f_{\epsilon}$ is concentrated near 0. We propose a data--driven bandwidth estimator and test its performance on \no exposure data.

\baselineskip=12pt
\par\vfill\noindent

\par\medskip\noindent
\underline{\bf Short title}: Kernel Density Estimation with Berkson Error

\clearpage\pagebreak\newpage
\pagenumbering{arabic}
\newlength{\gnat}
\setlength{\gnat}{22pt}
\baselineskip=\gnat
\section{Introduction} \label{sec:seca}
\label{sec:introduction}

\subsection{Background}
We consider smoothing a density estimate when an error--free sample is observed and one is interested in the convolution of the population density with an error term. This is known as density estimation with Berkson error and has been studied in \cite{delaigle2007nonparametric} where $NO_2$ exposure in children is estimated using known kitchen and bedroom concentrations. The exposure level in children is modeled as a function of kitchen and bedroom concentrations plus independent, additive random error.

Density estimation with Berkson error is one example of a class of problems where low--error or error--free data is used to construct an estimate for noisy data. \cite{bovy2011think} considers this problem in the context of constructing a classifier for noisy astronomical data. Here, each object belongs to the class quasar or star. For each object a telescope records a vector of flux ratios. For a training set of observations of known class, the authors observe low--error flux ratios. However for the data of unknown class, flux ratios contain significant measurement error. The goal is to construct an accurate classifier for the noisy data.\footnote{See Section 2 (Equations 1, 2, and 3) and Section 5 of \cite{bovy2011think} for more information.} \cite{carroll2009nonparametric} considers a similar problem in a regression context with data arising from nutritional epidemiology. \cite{long2012optimizing} studies this problem in the context of classification of periodic variable stars.

In each of these works, tuning parameters are selected to optimize some risk function. There is extensive literature on selecting tuning parameters for problems where all data is observed without measurement error. For example with kernel density estimation, asymptotic rates for the mean integrated squared error ($\MISE$) as a function of bandwidth as well as finite--sample procedures for selecting the bandwidth are known \citep{silverman1986density,jones1996brief}. In contrast, much less is understood about selection of tuning parameters when one set of data is error--free and there is measurement error in the variable of interest. In this work, we derive asymptotic results and present finite sample simulations illustrating the relationship between measurement error and smoothing for kernel density estimators. While our results are most directly applicable to density estimation with Berkson error, they have implications for the classification and regression problems above.

We now formalize the density estimation problem. Suppose we observe independent $\{X_i\}_{i=1}^n \sim f_X$. We seek to use this data to estimate the density, denoted $f_Y$, of $Y = X + \epsilon$. Here $\epsilon \sim f_{\epsilon}$, $X \sim f_{X}$, and $\epsilon$ and $X$ are independent. All random variables are in $\mathbb{R}^p$.  In the literature, $\epsilon$ is known as Berkson error and was introduced in a regression context by \cite{berkson1950there}. It differs from classical measurement error where one observes an error contaminated sample and seeks to estimate the underlying, uncontaminated density.\footnote{See \cite{carroll2006measurement} for a review of Berkson and classical measurement error.}

For estimating $f_Y$, \cite{delaigle2007nonparametric} proposed using
\begin{equation}
\label{eq:delaigle}
\widetilde{f}_Y(y) = \frac{1}{n}\sum_{i=1}^n f_{\epsilon}(y - X_i).
\end{equation}
\cite{delaigle2007nonparametric} showed that when $f_{\epsilon}$ is square--integrable and $f_X$ is bounded, this estimator is unbiased with a mean integrated squared error (MISE) that converges to 0 at rate $n$. The convergence result contrasts with standard density estimation where the $\MISE$ is generally of order $n^{-4/(4 + p)}$.\footnote{The $n^{-4/(4+p)}$ order for the $\MISE$ requires regularity conditions on $f_Y$. For example, \cite{wand1995kernel} (Section 4.3, \page 95) assumes each entry of the Hessian of $f_Y$ is piecewise continuous and square integrable. See \page 100 of \cite{wand1995kernel} for the $\MISE$ convergence rate.} Effectively, knowledge of $f_{\epsilon}$ provides valuable information about the structure of $f_Y$ unavailable in the standard kernel density estimation case. In addition to a fast convergence rate, $\widetilde{f}_Y$ in Equation \ref{eq:delaigle} has no tuning parameters that require estimation. A potential drawback of $\widetilde{f}_{Y}$ in Equation \ref{eq:delaigle} is that for cases where $\epsilon$ is concentrated around $0$, the estimator will have large spikes at the sample points $\{X_i\}_{i=1}^n$ and thus high variance. We illustrate this problem through simulations in Section \ref{sec:finite_sample}.

In this work we study of impact of smoothing estimates of $f_Y$ using kernels. We focus on comparing three approaches to kernel bandwidth selection:
\begin{itemize}
\item \textbf{Approach 1:} Select a bandwidth specifically to optimize estimation of $f_Y$.
\item \textbf{Approach 2:} Since $f_Y = \int f_X(y-\epsilon)f_{\epsilon}(\epsilon)d\epsilon$, use a kernel density estimator to estimate $f_X$. Select the bandwidth to optimize estimation of $f_X$. Then convolve this estimate of $f_X$ with $f_{\epsilon}$ in order to estimate $f_Y$.
\item \textbf{Approach 3:} Set the bandwidth to $0$ i.e., use Delaigle's estimator in Equation \eqref{eq:delaigle}.
\end{itemize}
Each of these approaches has some attractive properties. Approach 1 may provide the optimal performance because the bandwidth is chosen specifically to estimate $f_Y$. Approach 2 is attractive because there is extensive literature on selecting a bandwidth which optimizes estimation of $f_X$. Approach 3 is attractive because there are no tuning parameters to estimate and the resulting estimator has parametric first order convergence rates as shown by \cite{delaigle2007nonparametric}.

In addition to shedding light on the Berkson density estimation problem, a comparison of the performance of these approaches may be useful when considering how to regularize classification or regression methods when training data is error--free (or low--error) and data of unknown class has measurement error.

\subsection{Outline of Work and Summary of Findings}

In Section \ref{sec:setup} we present a kernel bandwidth estimator for $f_Y$ that allows simultaneous comparison of all three smoothing approaches. As a guide towards comparing the approaches, in Section \ref{sec:asymptotics} we derive a second order expansion of the $\MISE$ of the estimator as a function of bandwidth. The asymptotic expansion reveals interesting properties:
\begin{itemize}
\item Asymptotically, the optimal bandwidth for estimating $f_Y$ (approach 1) is of order $n^{-1/2}$. Notably this rate does not depend on the dimension of the problem, unlike for standard kernel density estimation. Approach 1 provides a second order ($n^{-2}$) reduction in $\MISE$ over approach 3 (no smoothing).
\item The asymptotically optimal bandwidth for estimating $f_Y$ depends on $f_X$ which is unknown. The expression we derive for the optimal bandwidth is used as the basis for a plug--in estimator in Section \ref{sec:real_data}.
\item Using approach 2 (smoothing to optimize estimation of $f_X$) results in a $\MISE$ of order $n^{-4/(4+p)}$, slower than the $n^{-1}$ order of approaches 1 and 3. As the dimension increases, the discrepancy in these rates grows.
\end{itemize}

Based purely on asymptotic considerations, approaches 1 and 3 appear superior to approach 2. In Section \ref{sec:finite_sample} we study the finite sample properties of these three approaches by adapting a result from \cite{wand1993comparison} which allows for exact computation of $\MISE$ when $f_X$ is a mixture of normals and the error and kernel are normal. We find:
\begin{itemize}
\item Approach 3 (no smoothing) can result in drastically undersmoothing the density, particularly when the error term is concentrated around $0$.
\item Approach 2 (smoothing to optimize estimation of $f_X$) can result in oversmoothing the density, particularly when the error term is smooth \ref{sec:asymptotics}.
\item In the cases considered in the simulation, the qualitative impacts of undersmoothing by using approach 3 appear worse than the qualitative impacts of oversmoothing by using approach 2.
\item Approach 1 outperforms approach 2 or 3. In the simulations considered, this performance advantage increased in three dimensions versus one dimension.
\end{itemize}

In Section \ref{sec:real_data} we propose a data--based bandwidth estimator for approach 1 and apply our methodology to the \no data studied by \cite{delaigle2007nonparametric}. In Section \ref{sec:conclusions} suggest some directions for future work. Proofs of all theorems are given in Section \ref{sec:proofs} and some technical issues are addressed in Section \ref{sec:technical_notes}.

\section{Problem Setup}
\label{sec:setup}
We observe independent random variables $X_1, \ldots, X_n \sim f_X$. We aim to estimate, $f_Y$, the density of
\begin{equation*}
Y = X + \epsilon.
\end{equation*}
Here $X\sim f_X$, $\epsilon \sim f_{\epsilon}$, and $X$ and $\epsilon$ are independent. $f_{\epsilon}$ is assumed known. All random variables are in $\mathbb{R}^p$. In all that follows let $\widehat{f}_V$ represent the characteristic function of the random variable $V$ and let $\widetilde{f}$ represent an estimator of $f$.

\subsection{Construction of an Estimator for $f_Y$}
We construct an estimator for $f_Y$ by first estimating $\widehat{f}_Y$, the characteristic function of $f_Y$. Let $K$ be a mean $0$ density function called the kernel, and $\widehat{K}$ its characteristic function. Let
\begin{equation*}
\Sigma_K = \int x x^T K(x)dx.
\end{equation*}
Let $H = H_n \succeq 0$ be a sequence of positive semidefinite $p \times p$ matrices called the bandwidth.
\begin{equation*}
\ffx = \frac{1}{n} \sum_{j=1}^n e^{i\omega^T X_j}
\end{equation*}
is an estimate of $\widehat{f}_X$. Consider estimating $\widehat{f}_Y$ (the characteristic function of $f_Y$) using
\begin{equation}
\label{eq:fourier_estimator}
\ffy = \widehat{K}(H\omega)\widehat{f}_{\epsilon}(\omega)\ffx.
\end{equation}
Note that $\ffyh$ is a characteristic function because it is the product of characteristic functions. Assuming $\ffyh \in L_1$, we may estimate $f_Y$ using
\begin{equation}
\label{eq:estimator}
\fy = \frac{1}{(2\pi)^p} \int e^{-i\omega^T y} \ffy d\omega. 
\end{equation}
The assumption that $\ffyh \in L_1$ implies $\fyh$ is a bounded density (see Theorem 3.3 in \cite{durrett2005probability}). Throughout this work, we require $\widehat{f}_{\epsilon} \in L_1$, thus guaranteeing that $\ffyh \in L_1$ and ensuring that $\fyh$ in Equation \eqref{eq:estimator} is a valid density.

\subsection{$\fyh$ and Approaches to Smoothing}
$\fyh$ allows for simultaneous comparison of all three bandwidth selection approaches discussed in the introduction. With approach 1, we optimize $H$ in $\fyh$ specifically for estimating $f_Y$. Recall that with approach 2, we construct a kernel density estimator for $f_X$ with some bandwidth $H_X \succ 0$ optimized for estimating $f_X$. Denote this estimator
\begin{equation*}
\fxhx{H_X}{x} = \frac{1}{n}\sum_{i=1}^n K_{H_X}(x - X_i).
\end{equation*}
Since $f_Y(\epsilon) = \int f_X(y - \epsilon)f_{\epsilon}(\epsilon)d\epsilon$, $\fxhx{H_X}{x}$ is convolved with $f_{\epsilon}$ to estimate $f_Y$. However this procedure is equivalent to using $H_X$ in $\fyh$. In other words,
\begin{equation*}
\fyhy{H_X}{y} = \int \fxhx{H_X}{y - \epsilon}f_{\epsilon}(\epsilon)d\epsilon.
\end{equation*}
Finally $\fyh$ includes as a subcase approach 3, no smoothing. By setting $H=0$ we have
\begin{equation*}
\fyhy{0}{y} = \frac{1}{n} \sum_{i=1}^n f_{\epsilon}(y - X_i).
\end{equation*}
This is the kernel--free estimator of \cite{delaigle2007nonparametric} presented in Equation \eqref{eq:delaigle}.

\subsection{MISE}

We use mean integrated squared error as a guide toward understanding the behavior of the three approaches to selecting the bandwidth in $\fyh$. Let $\mathcal{P}_n$ be the product measure on $(X_1, \ldots, X_n)$. The mean integrated squared error is defined as
\begin{equation*}
\MISE(H) \equiv \mathbf{E}_{\mathcal{P}_n}\int \left(\fy - f_Y(y)\right)^2 dy.
\end{equation*}
$\MISE$ is a popular measure of risk used in many density estimation studies (see for example \cite{jones1996brief,marron1992exact,delaigle2008alternative,wand1995kernel,tsybakov2009introduction}). The $\MISE$ allows for relatively straightforward asymptotic analysis (Theorems \ref{thm:char_mise} and \ref{thm:multi_dim}) as well as admitting to, under certain conditions, a computationally efficient representation which we exploit in order to obtain our finite sample results (Theorem \ref{thm:normal_mise}). We explore some of the qualitative impacts of the different smoothing approaches, not directly captured by $\MISE$, in Subsection \ref{sec:conf_bounds}. While other measures, such as K-L divergence, Hellinger distance, or integrated absolute error, could be used, we feel $\MISE$ captures the essential properties of the three approaches to smoothing.

The optimal $H$, in terms of minimizing $\MISE$ for estimating $f_Y$, is
\begin{equation*}
H_{Y} = \argmin{\{H: H \succeq 0\}} \MISE(H).
\end{equation*}
This is the bandwidth of approach 1. For approach 2, the bandwidth is chosen to minimize the $\MISE$ in estimating $f_X$. In other words
\begin{equation*}
H_X = \argmin{\{H: H \succ 0\}} \mathbf{E}_{\mathcal{P}_n}\int \left(\widetilde{f}_X(x) - f_X(x)\right)^2 dx.
\end{equation*}
Unfortunately the $\MISE$ expression is complicated and exact expressions for $H_Y$ and $H_X$ are not generally possible. In Section \ref{sec:asymptotics} we form asymptotic approximations to the $\MISE$ and determine the rates at which $||H_{Y}||_{\infty} \rightarrow 0$ and $\MISE(H_Y) \rightarrow 0$ as $n \rightarrow \infty$. We refer to standard kernel density theory for the rates associated with $H_X$. Our asymptotic expansion of the $\MISE$ reveals rates associated with $\MISE(H_X)$. The asymptotic approximation also shows the error for approach 3, $\MISE(0)$, a quantity that was derived by \cite{delaigle2007nonparametric}.

In Section \ref{sec:finite_sample} we specialize to the case where $f_X$ is a Gaussian mixture, $K$ is a Gaussian kernel, and $f_{\epsilon}$ is a Gaussian error density. In this setting, the $\MISE$ can be evaluated at a particular $H$ without numerically approximating integrals (see Theorem \ref{thm:normal_mise}). Using this result we study the finite sample properties of $H_Y$, $H_X$, $\MISE(H_Y)$, $\MISE(H_X)$, and $\MISE(0)$.

In this work, $H_Y$ refers to the exact optimal bandwidth, $H_Y^*$ refers to an asymptotically optimal bandwidth, and $\widetilde{H}_Y$ refers to a data based estimator of $H_Y$ (used mostly in Section \ref{sec:real_data}). While we derive asymptotic results for bandwidth matrices, we often specialize to the case of a scalar bandwidth. In such cases, lowercase letters, $h_Y$, $h_Y^*$, and $\widetilde{h}_Y$, are used. The same notation applies for $H_X$.

\section{Asymptotic Results}
\label{sec:asymptotics}

For the purposes of forming asymptotic expansions, we represent the $\MISE$ in terms of characteristic functions.
\begin{Thm}
\label{thm:char_mise}
Assume $\widehat{f}_Y, \, \ffyh \in L_1$. Then
\begin{equation}
\label{eq:misefourier}
(2\pi)^p \MISE(H) = \int |1 - \widehat{K}(H\omega)|^2 d\mu(\omega) + \frac{1}{n}\int |\widehat{K}(H\omega)|^2 d\nu(\omega)
\end{equation}
where
\begin{align*}
&d\mu(\omega) = |\widehat{f}_{\epsilon}(\omega)|^2|\widehat{f}_X(\omega)|^2 d\omega,\\
&d\nu(\omega) = |\widehat{f}_{\epsilon}(\omega)|^2 (1-|\widehat{f}_X(\omega)|^2) d\omega
\end{align*}
are positive measures.
\end{Thm}
See Subsection \ref{sec:char_mise} on \page \pageref{sec:char_mise} for a proof. The representation of the $\MISE$ in Equation \eqref{eq:misefourier} closely resembles that of \cite{tsybakov2009introduction} Theorem 1.4. In Equation \eqref{eq:misefourier}, $\int |1 - \widehat{K}(H\omega)|^2 d\mu(\omega)$ is the integrated squared bias of $\fyh$ and $n^{-1}\int |\widehat{K}(H\omega)|^2 d\nu(\omega)$ is the integrated variance. Notice that for fixed $H$, the variance decreases at rate $n^{-1}$ while the bias is constant. When $H=0$, $\widehat{K}(H\omega) = 1$, so the integrated squared bias term vanishes.

We require assumptions on the kernel $K$ and the bandwidth matrix $H$.
\begin{Assume}
\label{assump:kernel}
\begin{align}
&K \text{ is a symmetric density} \label{assump:symmetric_density}\\
&\widehat{K} \text{ is four times continuously differentiable } \label{assump:1kernel}\\
&H = H_n \succeq 0 \text{ (i.e. sequence is positive semidefinite)} \label{assump:pos_semi}\\
&||H||_{\infty} \rightarrow 0
\end{align}
\end{Assume}
Since we choose the kernel and bandwidth matrix, these assumptions can always be satisfied in practice. Common kernel choices such as the standard normal and uniform on $[-1,1]^p$ satisfy Assumptions \ref{assump:symmetric_density} and \ref{assump:1kernel}. In Assumption \ref{assump:pos_semi}, notice that the bandwidth does not have to be strictly positive definite, unlike in the standard kernel density estimation case. We require the following assumptions on the characteristic functions of $f_X$ and $f_{\epsilon}$. 
\begin{Assume}
\label{assump:density}
\begin{align}
&\int ||\omega||_{\infty}^8 |\widehat{f}_{\epsilon}(\omega)|^2d\omega < \infty \label{assump:1ep}\\
&\int |\widehat{f}_{\epsilon}(\omega)| d\omega < \infty \label{assump:ep_l1}
\end{align}
\end{Assume}
Assumptions \ref{assump:1ep} and \ref{assump:ep_l1} are satisfied as long as the error term has a density that is smooth, such as multivariate normal or Student's t (see \cite{sutradhar1986characteristic} for the characteristic function of the multivariate Student's t).
\begin{Thm}
\label{thm:multi_dim}
Under Assumptions \ref{assump:kernel} and \ref{assump:density} and with the notation of Theorem \ref{thm:char_mise}
\begin{align}
\label{eq:full_bandwidth_mise}
&(2\pi)^p \MISE(H) \nonumber \\
&= \frac{1}{n}\int d\nu(\omega) \nonumber \\
&+ \left(\frac{1}{4} \int (\omega^TH^T \Sigma_K H\omega)^2d\mu(\omega) - \frac{1}{n}\int (\omega^T H^T \Sigma_K H \omega) d\nu(\omega)\right)(1 + O(||H||_{\infty}^2)).
\end{align}
\end{Thm}
See Subsection \ref{sec:multi_d} on \page \pageref{sec:multi_d} for a proof. The $n^{-1}\int d\nu(\omega)$ term is variance in the estimator that does not depend on the bandwidth. If the bandwidth is set to $0$ (see approach 3 below), all other terms vanish and this is the $\MISE$. The $\frac{1}{4} \int (\omega^TH^T \Sigma_K H\omega)^2d\mu(\omega)$ term is bias caused by using a kernel with bandwidth $H$ while $-n^{-1}\int (\omega^T H^T \Sigma_K H \omega) d\nu(\omega)$ is the corresponding reduction in variance. 


While the full bandwidth matrix offers the most flexibility and greatest potential for reduction in MISE, this expression is difficult to optimize, see Subsection \ref{sec:full_bandwidth}. More simply, one could use a diagonal bandwidth matrix and optimize $p$ bandwidths. We study this case in Subsection \ref{sec:diagonal}. Here we focus on using a scalar bandwidth. This allows for the most direct and straightforward comparisons of the different approaches to smoothing.

\subsection{Scalar Bandwidth}
\label{sec:scalar_bandwidth_density}
We reparameterize the bandwidth $H = hI$. Here the general $\MISE$ expression in Equation \eqref{eq:full_bandwidth_mise} becomes
\begin{align}
&(2\pi)^p \MISE(h) \nonumber \\
&= \frac{1}{n}\int d\nu(\omega) + \left(\frac{h^4}{4} \int (\omega^T \Sigma_K \omega)^2 d\mu(\omega) - \frac{h^2}{n}\int (\omega^T \Sigma_K \omega) d\nu(\omega) \right)(1 + O(h^2)). \label{eq:scalar_mise_degenerate}
\end{align}
We now discuss the three smoothing approaches

\textbf{Approach 1:} It is straightforward to find the bandwidth that minimizes the main terms in this $\MISE$ expression. Specifically,
\begin{align}
h^*_Y &= \argmin{h \geq 0} \left(\frac{h^4}{4} \int (\omega^T \Sigma_K \omega)^2 d\mu(\omega) - \frac{h^2}{n}\int (\omega^T \Sigma_K \omega) d\nu(\omega) \right) \nonumber \\
&= \sqrt{\frac{2\int (\omega^T \Sigma_K \omega) d\nu(\omega)}{n\int (\omega^T \Sigma_K \omega)^2 d\mu(\omega)}}. \label{eq:1dasymptopt}
\end{align}
$h^*_Y$ is of order $n^{-1/2}$. This order does not depend on $p$, unlike for the error free case where, under some regularity conditions, the order for the asymptotically optimal amount of smoothing is $n^{-1/(4+p)}$. However, the dimension $p$ will affect the constant on $h_Y^*$ in Equation \eqref{eq:1dasymptopt}. We explore the relationship between the dimension $p$ and the optimal amount of smoothing for finite samples in Section \ref{sec:finite_sample}.  Using, $h^*_Y$ the $\MISE$ is
\begin{equation*}
(2\pi)^p \MISE(h^*_Y) = \frac{1}{n}\int d\nu(\omega) \\
- \, \frac{1}{n^2} \frac{\left(\int (\omega^T \Sigma_K \omega) d\nu(\omega) \right)^2}{\left(\int (\omega^T \Sigma_K \omega)^2 d\mu(\omega)\right)} + O(n^{-3}).
\end{equation*}

\textbf{Approach 2:}
Set $h$ to minimize $\MISE$ in estimating $f_X$. Under certain regularity conditions on $f_X$, the bandwidth is order $n^{-1 / (4 + p)}$ (e.g. see \cite{wand1995kernel} page 100). Specifically, suppose
\begin{equation*}
h_X^* = D(n)n^{-1/(4+p)},
\end{equation*}
where $D: \mathbb{Z}^+ \rightarrow \mathbb{R}^+$ such that $\limsup_n D(n) < \infty$ and $\liminf_n D(n) > 0$. The $\MISE$ for estimating $f_Y$ using $h_X^*$ (obtained from Equation \eqref{eq:scalar_mise_degenerate}) is 
\begin{align}
(2\pi)^p \MISE(h_X^*)=&\frac{1}{n}\int d\nu(\omega) + \bigg(\frac{D(n)^4n^{-4/(4+p)}}{4} \int (\omega^T \Sigma_K \omega)^2 d\mu(\omega) \nonumber \\
&- D(n)^2n^{-(6+p)/(4+p)}\int (\omega^T \Sigma_K \omega) d\nu(\omega) \bigg)(1 + O(n^{-2/(4+p)})) \nonumber \\
&= \left(\frac{D(n)^4n^{-4/(4+p)}}{4} \int (\omega^T \Sigma_K \omega)^2 d\mu(\omega)\right)(1 + o(1)). \label{eq:fob}
\end{align}
The $n^{-4/(4+p)}$ order for the $\MISE$ when using $h_X^*$ is strictly worse than the $n^{-1}$ order that can be achieved by optimizing the bandwidth specifically for the error distribution, i.e. using $h^*_Y$. Essentially using $h_X^*$ oversmooths $\fyh$. The first order term in $\MISE(h_X^*)$, Equation \eqref{eq:fob}, is caused entirely by bias.

\textbf{Approach 3:}
Set $h=0$. Here we have
\begin{equation*}
(2\pi)^p\MISE(0) = \frac{1}{n}\int d\nu(\omega) = \frac{1}{n}\left(\int |\widehat{f}_{\epsilon}(\omega)|^2d\omega - \int |\widehat{f}_{\epsilon}(\omega)|^2|\widehat{f}_X(\omega)|^2d\omega\right).
\end{equation*}
Asymptotically, this approach is better than approach 2 since $\MISE(0)$ is order $n^{-1}$. The ratio of using asymptotically optimal smoothing ($h_Y^*$) to no smoothing is
\begin{equation*}
\frac{\MISE(h_Y^*)}{\MISE(0)} = 1 - \frac{1}{n}\frac{\left(\int (\omega^T \Sigma_K \omega) d\nu(\omega) \right)^2}{\left(\int (\omega^T \Sigma_K \omega)^2 d\mu(\omega)\right) \left( \int d\nu(\omega) \right)} + O(n^{-2}).
\end{equation*}

These asymptotic results suggest that using the bandwidth that minimizes error in estimating $f_X$ for estimating $f_Y$ is a poor idea. This procedure results in a convergence rate of higher order than either not smoothing (approach 3) or smoothing specifically for $f_Y$ (approach 1), i.e. using $h^*_Y$. The asymptotic results show an improvement only at the $n^{-2}$ level by using $h=h^*_Y$ rather than $h=0$. Based strictly on asymptotic analysis, this small improvement in error rate may not appear to justify the extra effort required to estimate $h_Y$. However, simulation results in Section \ref{sec:finite_sample} indicate that the effects of using $h_Y$ are more important than the asymptotic analysis suggest, both in terms of minimizing $\MISE$ and preserving important qualitative features of the densities.

\section{Finite Sample Results}
\label{sec:finite_sample}

The asymptotic results from Section \ref{sec:asymptotics} illustrate the large sample behavior of the $\MISE$ under different approaches to choosing the bandwidth parameter. However it is important to understand the finite sample behavior of these quantities and what sample sizes are needed for asymptotics to be informative.

Calculation of the exact $\MISE$ in Equation \ref{eq:misefourier} requires numerically approximating integrals. Therefore it is computationally challenging, given a $f_X$, $f_{\epsilon}$, and kernel $K$, to determine the bandwidth $H$ which minimizes the $\MISE$. For the error-free kernel density estimation case, \cite{wand1993comparison} showed that when $f_X$ is a normal mixture and $K$ is a normal kernel, the exact $\MISE$ has a simple representation that does not require numerically approximating integrals. Using this result, the authors compared bandwidth parameterizations for bivariate density estimation. Here we generalize this result to the case with Berkson measurement error. We assume $f_X$ is a normal mixture and $K$ and $f_{\epsilon}$ are normal. We use this result for studying the finite sample properties of various bandwidth selection approaches. Theorem \ref{thm:normal_mise} is a generalization of Theorem 1 in \cite{wand1993comparison} to the case with Berkson error.
\begin{Thm}
\label{thm:normal_mise}
Let $\phi_\Sigma$ be the mean $0$, normal density with covariance $\Sigma$. Assume the kernel $K = \phi_{\Sigma_K}$ and the error density $f_{\epsilon} = \phi_{\Sigma_{\epsilon}}$. Assume that $f_X$ is a mixture of normal densities parameterized by $\{(\w_j,\mu_j,\Sigma_j)\}_{j=1}^m$ where $\sum_{j=1}^m \w_j = 1$ and $(\w_j,\mu_j,\Sigma_j)$ is the mixing proportion, mean, and variance of the $j^{th}$ component of $f_X$. In other words
\begin{equation*}
f_X(x) = \sum_{j=1}^m \w_j \phi_{\Sigma_j}(x - \mu_j).
\end{equation*}
Let $S = H^T \Sigma_K H$. Let $\Omega_a$ for $a \in \{0,1,2\}$ be a $m \times m$ matrix with $j,j'$ entry equal to 
\begin{equation*}
\phi_{aS + 2\Sigma_{\epsilon} + \Sigma_{j} + \Sigma_{j'}}(\mu_j - \mu_{j'}).
\end{equation*}
Finally let $\w = (\w_1, \ldots, \w_m)$. Then
\begin{equation}
\label{eq:exact_mise}
\MISE(H) = \frac{1}{n} \phi_{2S + 2\Sigma_\epsilon}(0) + \w^T ((1 - n^{-1}) \Omega_2 - 2\Omega_1 + \Omega_0)\w.
\end{equation}
\end{Thm}
See Subsection \ref{sec:normal_mise_proof} for a proof. Equation \eqref{eq:exact_mise} can be evaluated at a particular bandwidth $H$ without numerically approximating integrals.

In Subsection \ref{sec:relative_error} we compare the $\MISE$ for the three bandwidth selection approaches at finite sample sizes in one dimension. In Subsection \ref{sec:relative_error_3d} we repeat this analysis for several 3--dimensional densities. In Subsection \ref{sec:conf_bounds} we show the visual impacts of different $\MISE$s by plotting pointwise quantiles for density estimates using different bandwidth selection approaches. Finally in Subsection \ref{sec:conv_rate} we explore how fast the asymptotic approximations from Section \ref{sec:asymptotics} for the optimal smoothing parameter for $f_Y$ take hold. All of the results presented in this section can be reproduced using publicly available code.\footnote{R-code and data for generating results in Sections \ref{sec:finite_sample} and \ref{sec:real_data} are available at \url{http://stat.tamu.edu/~jlong/berkson.zip}.}

\subsection{Relative Error in One Dimension}
\label{sec:relative_error}

In this section the $f_X$ densities are 1--dimensional normal mixtures which satisfy the assumptions of Theorem \ref{thm:normal_mise}. The densities we consider, and associated names, are plotted in Figure \ref{fig:densities}. The exact parameter values for these densities are given in Table \ref{tab:densities}. 

\begin{figure}[h]
  \begin{center}
    \begin{includegraphics}[height=5.5in,width=5in]{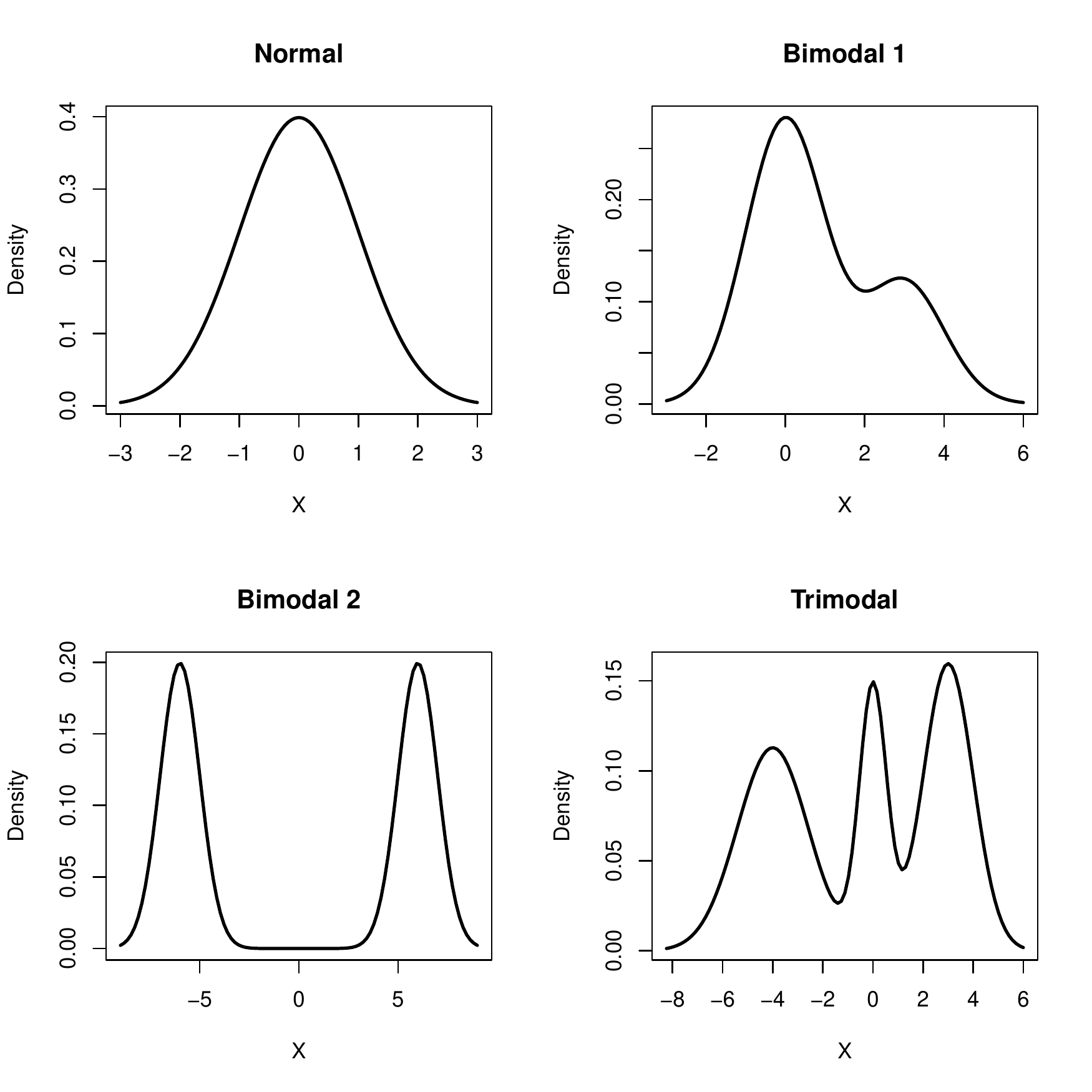}
      \caption{The four 1--dimensional Gaussian mixture densities we consider. \label{fig:densities}}
    \end{includegraphics}
  \end{center}
\end{figure}

\begin{table}[ht]
\centering
\begin{tabular}{r|cc}
  \hline
Name & Parameters \\
  \hline
Normal & $\phi_1(x)$\\
Bimodal 1 & $.7\phi_1(x) + .3\phi_1(x-3)$\\
Bimodal 2 &  $.5\phi_1(x+6) + .5\phi_1(x-6)$\\
Trimodal &  $.4\phi_2(x+4) + .2\phi_{.3}(x) + .4\phi_1(x-3)$\\
   \hline
\end{tabular}
\caption{Parameters for the four densities plotted in Figure \ref{fig:densities} where $\phi_\Sigma$ is the normal, mean 0 density with covariance $\Sigma$.} 
\label{tab:densities}
\end{table}

We now compare the $\MISE$ for the densities in Figure \ref{fig:densities} using the three approaches for selecting the bandwidth parameter. Recall that the approaches are: 1) optimize the bandwidth for estimating $f_Y$, 2) optimize the bandwidth for estimating $f_X$, and 3) set the bandwidth equal to $0$. In Section \ref{sec:asymptotics} we showed that optimizing the bandwidth for $f_Y$ (approach 1) and setting the bandwidth equal to $0$ (approach 3) resulted in the same first order asymptotic performance for the $\MISE$. In contrast, optimizing the bandwidth for $f_X$ (approach 2) resulted in slower, nonparametric convergence rates.

Let $h_Y$ be the optimal bandwidth for estimating $f_Y$ and $h_X$ be the optimal bandwidth for estimating $f_X$ ($h_Y$ and $h_X$ are sequences implicitly indexed by the sample size $n$). We now compare $\MISE(h_Y)$, $\MISE(h_X)$, and $\MISE(0)$ at finite $n$ for the four densities in Figure \ref{fig:densities} and a variety of error variances $\sigma_{\epsilon}^2$. Clearly $\MISE(h_Y) \leq MISE(h_X)$ and $\MISE(h_Y) \leq MISE(0)$ since $h_Y$ is the minimizer of the $\MISE$. We seek to understand the level of reduction in $\MISE$ one can achieve by using $h_Y$, the parameter settings where these reductions occur, and how $\MISE(h_X)$ compares to $\MISE(0)$. We note that $h_X$ and $h_Y$ are exact minimizers for the $\MISE$ of $f_Y$ and $f_X$, not asymptotic approximations.

In Table \ref{tab:50} we present $\left(\frac{MISE(0)}{MISE(h_Y)},\frac{MISE(h_X)}{MISE(h_Y)}\right)$ for four densities and five error variances (the error is normal, mean 0) for $n=50$. We note some general trends. As $\sigma_{\epsilon}^2$ decreases, $\MISE(h_X)/MISE(h_Y)$ decreases. With small $\sigma_{\epsilon}^2$, $f_Y$ is close to $f_X$ and thus $h_X$ and $h_Y$ are close. In contrast, as $\sigma_{\epsilon}^2$ decreases, $\MISE(0)/MISE(h_Y)$ increases. With no smoothing and small $\sigma_{\epsilon}^2$, approach 3 undersmooths the density estimate.

\begin{table}[ht]
\centering
\begin{tabular}{r|cccc}
  \hline
$\sigma_{\epsilon}^2$ & Normal & Bimodal 1 & Bimodal 2 & Trimodal \\ 
  \hline
2 & (1.02,1.18) & (1.08,1.01) & (1.03,1.02) & (1.18,1.05) \\ 
  1 & (1.05,1.17) & (1.15,1.01) & (1.07,1.03) & (1.24,1.04) \\ 
  0.5 & (1.13,1.11) & (1.26,1.01) & (1.16,1.03) & (1.30,1.01) \\ 
  0.25 & (1.32,1.05) & (1.50,1.00) & (1.37,1.01) & (1.46,1.00) \\ 
  0.125 & (1.70,1.02) & (1.92,1.00) & (1.76,1.01) & (1.77,1.00) \\ 
   \hline
\end{tabular}
\caption{Each entry is $\left(\frac{MISE(0)}{MISE(h_Y)},\frac{MISE(h_X)}{MISE(h_Y)}\right)$ for $n=50$. These ratios are always greater than $1$ because $h_Y$ is the minimizer of the $\MISE$. As expected, $\MISE(0)$ performs well when $\sigma_{\epsilon}^2$ (the error variance) is large but poorly when $\sigma_{\epsilon}^2$ is small. $\MISE(h_X)$ performs well when $\sigma_{\epsilon}^2$ is small but poorly when $\sigma_{\epsilon}^2$ is large.} 
\label{tab:50}
\end{table}

For the densities and error variances considered, using $h_X$ (approach 2) is generally better than no smoothing (approach 3). Only for the normal distribution with $\sigma_{\epsilon}^2=2$ or $1$ does no smoothing outperform smoothing with $h_X$. This is surprising given the asymptotic results showed that the convergence rate using $h_X$ is slower than the rate using no smoothing. For the densities and error distributions considered, a sample size of $50$ is not large enough for these asymptotics to take hold. An important caveat to this conclusion is that the no smoothing estimator is simpler than $h_X$ because it has no smoothing parameters.


\begin{table}[ht]
\centering
\begin{tabular}{r|cccc}
  \hline
$\sigma_{\epsilon}^2$ & Normal & Bimodal 1 & Bimodal 2 & Trimodal \\ 
  \hline
2 & (1.01,1.24) & (1.04,1.03) & (1.02,1.04) & (1.09,1.02) \\ 
  1 & (1.03,1.24) & (1.08,1.03) & (1.04,1.06) & (1.12,1.01) \\ 
  0.5 & (1.07,1.18) & (1.15,1.03) & (1.09,1.06) & (1.16,1.00) \\ 
  0.25 & (1.19,1.09) & (1.31,1.02) & (1.24,1.03) & (1.27,1.00) \\ 
  0.125 & (1.46,1.04) & (1.62,1.01) & (1.53,1.01) & (1.50,1.00) \\ 
   \hline
\end{tabular}
\caption{The entries here are the same as Table \ref{tab:50} but for $n=100$. This larger $n$ generally improves performance for $\MISE(0)$ and worsens the performance of $\MISE(h_X)$ (relative to $\MISE(h_Y)$). This is predicted by our asymptotic theory, since as $n \rightarrow \infty$, $\frac{MISE(0)}{MISE(f_Y)} \rightarrow 1$ while $\frac{MISE(h_X)}{MISE(h_Y)} \rightarrow \infty$. However at $n=100$, using $h_X$ still generally outperforms no smoothing.} 
\label{tab:100}
\end{table}

In Table \ref{tab:100} we plot the same quantities as Table \ref{tab:50} but for $n=100$. The same general trends apply here as with the $n=50$ case: As $\sigma_{\epsilon}^2$ decreases $\MISE(h_X)$ decreases while $\MISE(0)$ increases (relative to $\MISE(h_Y)$). Note that the performance of no smoothing relative to $h_X$ is generally better for $n=100$ than $n=50$. For example, with Bimodal 2 and $\sigma_{\epsilon}^2 = 2$, $\MISE(h_X) < MISE(0)$ for $n=50$, but $\MISE(h_X) > MISE(0)$ for $n=100$. In fact, for every case considered $\MISE(0)/MISE(h_Y)$ is lower for $n=100$ than $n=50$. With the exception of Trimodal, $\MISE(h_X)/MISE(h_Y)$ is always greater for $n=100$ than $n=50$.

The asymptotic results in Section \ref{sec:scalar_bandwidth_density} predict this behavior. As $n \rightarrow \infty$, $\frac{MISE(0)}{MISE(f_Y)} \rightarrow 1$ while $\frac{MISE(h_X)}{MISE(h_Y)} \rightarrow \infty$. So for large enough $n$, $\MISE(0) < MISE(h_X)$. However, for most densities and error variances considered here, $n=100$ is not large enough for these asymptotics to take hold. These results suggest that at sample sizes of potential interest, using no smoothing can undersmooth the density estimate. This effect appears most significant when the error density is concentrated near 0.

\subsection{Relative Error in Three Dimensions}
\label{sec:relative_error_3d}

We now explore relative error rates for 3--dimensional densities. The Gaussian mixture densities we study are defined in Table \ref{tab:multi_densities} using the notation
\begin{equation}
\label{eq:plusminus}
+ = \begin{bmatrix}
1 & 0.64 & 0\\
0.64 & 1 & 0.64\\
0 & 0.64 & 1\\
\end{bmatrix}
\, \, \, \, \, \, \, \, \, \, \, \, \, \, 
- = \begin{bmatrix}
1 & -0.64 & 0\\
-0.64 & 1 & -0.64\\
0 & -0.64 & 1\\
\end{bmatrix}
\end{equation}
for covariances matrices. The four densities in Table \ref{tab:multi_densities} are meant to be 3--dimensional generalizations of the 1--dimensional densities from Table \ref{tab:densities}. In particular Multi. Normal is a 3--dimensional normal, a direct generalization of the 1--dimensional normal. Multi 2-Comp 2 is a mixture of two well separated, identity covariance normals. This is a close analogue to Bimodal 2 from Table \ref{tab:densities}.
\begin{table}[ht]
\centering
\begin{tabular}{r|cc}
  \hline
Name & Parameters \\
  \hline
Multi. Normal & $\phi_I(x)$\\
Multi. 2-Comp 1 & $.7\phi_{+}(x) + .3\phi_{-}(x-(1,1,1)^T)$\\
Multi. 2-Comp 2 &  $.5\phi_I(x-(6,0,0)^T) + .5\phi_I(x+(6,0,0)^T)$\\
Multi. 3-Comp &  $.4\phi_{+}(x) + .2\phi_{-}(x-(1,1,1)^T) + .4\phi_-(x)$\\
   \hline
\end{tabular}
\caption{Parameters for the four 3--dimensional densities studied. $\phi_\Sigma$ is the normal, mean 0 density in three dimensions with covariance $\Sigma$. Here $I$ is the $3 \times 3$ identity matrix and the $+$ and $-$ signs are covariance matrices defined in Equation \ref{eq:plusminus}.}
\label{tab:multi_densities}
\end{table}

We consider 5 error densities for $\epsilon$. Each error density is normal, mean 0, with all diagonal elements equal and 0 for all covariances. The normality of $\epsilon$ is required by Theorem \ref{thm:normal_mise}. The other choices were made to keep these simulations a reasonable size. For diagonal terms of the covariance, we consider the same values as for the 1--dimensional case: $2,1,0.5,0.25,$ and $0.125$.

Tables \ref{tab:multi100} and \ref{tab:multi500} present the ratios $\left(\frac{MISE(0)}{MISE(h_Y)},\frac{MISE(h_X)}{MISE(h_Y)}\right)$ for all 20 $f_X$, $f_{\epsilon}$ pairs for $n=100$ and $n=500$ respectively. The first column in each table, $\sigma_\epsilon^2$, refers to the diagonal elements of the covariance matrix. The $n=100$ case, Table \ref{tab:multi100}, allows for direct comparison with the 1--dimensional case in Table \ref{tab:100}. The $n=500$ case, Table \ref{tab:multi500}, provides results for what is perhaps a more realistic sample size when attempting to estimate a 3--dimensional density non--parametrically.

We discuss the $n=100$ results, Table \ref{tab:multi100}. In general $h_Y$ performs better relative to no smoothing and $h_X$ smoothing in three dimensions than in one dimension. In particular, all ratios are larger for Multi. Normal and Multi. 2-Comp 2 in Table \ref{tab:multi100} than Normal and Bimodal 2 in Table \ref{tab:100}. As before, with small error variance, approach 3 undersmooths the density estimate. As in the 1-dimensional case, $h_X$ oversmooths the density estimates when the error variance is large. This effect appears worse in three dimensions than in one dimension. For example, for the standard normal with $\sigma_{\epsilon}^2 = 2$ and $n=100$, $\MISE(h_X) / MISE(h_Y)$ is 1.24 in one dimension (Table \ref{tab:100}) and 1.76 in three dimensions (Table \ref{tab:multi100}).

\begin{table}[ht]
\centering
\begin{tabular}{r|cccc}
  \hline
$\sigma_{\epsilon}^2$ & Multi Normal & Multi 2-Comp 1 & Multi 2-Comp 2 & Multi 3-Comp \\ 
  \hline
2 & (1.02,1.76) & (1.02,1.13) & (1.04,1.28) & (1.02,1.20) \\ 
  1 & (1.07,1.63) & (1.06,1.12) & (1.12,1.28) & (1.07,1.15) \\ 
  0.5 & (1.24,1.35) & (1.16,1.08) & (1.39,1.17) & (1.21,1.07) \\ 
  0.25 & (1.80,1.14) & (1.40,1.05) & (2.18,1.07) & (1.55,1.02) \\ 
  0.125 & (3.38,1.05) & (2.00,1.02) & (4.34,1.02) & (2.32,1.01) \\ 
   \hline
\end{tabular}
\caption{Three dimensional finite sample results for $n=100$. Generally, $h_X$ and no smoothing perform worse relative to $h_Y$ here than for $n=100$ in one dimension (see Table \ref{tab:100}).} 
\label{tab:multi100}
\end{table}

We now discuss the results for $n=500$, Table \ref{tab:multi500}. Note that with the larger sample size, all of the $\MISE(0)/MISE(h_Y)$ ratios have decreased while all of the $\MISE(h_X)/MISE(h_Y)$ ratios have increased relative to the $n=100$ case in Table \ref{tab:multi100}. The asymptotic results from Section \ref{sec:asymptotics} predict this behavior. As $n \rightarrow \infty$, $\MISE(0)/MISE(h_Y) \rightarrow 1$ while  $\MISE(h_X)/MISE(h_Y) \rightarrow \infty$. As expected, for all densities $\MISE(0)/MISE(h_Y)$ is increasing as $\sigma_\epsilon^2$ decreases.

\begin{table}[ht]
\centering
\begin{tabular}{r|cccc}
  \hline
$\sigma_{\epsilon}^2$ & Multi Normal & Multi 2-Comp 1 & Multi 2-Comp 2 & Multi 3-Comp \\ 
  \hline
2 & (1.00,2.66) & (1.00,1.27) & (1.01,1.72) & (1.01,1.37) \\ 
  1 & (1.01,2.54) & (1.01,1.30) & (1.03,1.82) & (1.02,1.34) \\ 
  0.5 & (1.06,2.00) & (1.03,1.29) & (1.10,1.57) & (1.05,1.25) \\ 
  0.25 & (1.25,1.45) & (1.10,1.23) & (1.41,1.26) & (1.14,1.16) \\ 
  0.125 & (1.94,1.16) & (1.30,1.14) & (2.37,1.09) & (1.41,1.09) \\ 
   \hline
\end{tabular}
\caption{Three dimensional finite sample results for $n=500$. $\MISE(h_X) / MISE(h_Y)$ is larger and $\MISE(0)/MISE(h_Y)$ is smaller here relative to Table \ref{tab:multi100} where the sample size was $100$.} 
\label{tab:multi500}
\end{table}

The three dimensional finite sample results reinforce the conclusion from the one dimensional results that not smoothing when the error variance is small produces undersmoothed estimates with large $\MISE$ relative to using the optimal smoothing parameter. Additionally, in three dimensions, using $h_X$ oversmooths the density estimates in many cases where the error variance is large.

The 3--dimensional simulations here were limited to a very narrow set of error distributions. In particular the error variance was equal in all directions. Another interesting setting to consider is when the error variance is very small along certain directions and sizable along other directions. The limiting case of this setting has no error along certain directions.

When there is no error in certain direction but error in other directions, one can obtain convergence rates that depend on the dimension of the space on which there is error (\cite{long2013thesis}, Theorem 2.3). Specifically, suppose one estimates a $p$ dimensional density $f_Y$ and $\epsilon$ is 0 with probability $1$ on a $p_0$ dimensional subspace of $\mathbb{R}^p$. Suppose $\epsilon$ has a density on the other $p_1$ dimensions ($p = p_0 + p_1$). Then for second order kernels, under some regularity conditions, the optimal smoothing using a scalar bandwidth is of order $n^{-1/(4 + p_1)}$ and results in an $\MISE$ of order $n^{-4/(4+p_1)}$ (see Equation 2.27 in \cite{long2013thesis}). Note that this rate is between the error free rate of $n^{-4/(4+p)}$ and the error in all directions case where the $\MISE$ is of order $n^{-1}$.

\subsection{Qualitative Impacts of Smoothing Parameters}
\label{sec:conf_bounds}

We now visualize some of the results in Table \ref{tab:50} by plotting pointwise quantiles for density estimates for different choices of smoothing parameters. In order to obtain an understanding of the impact of using $h_X$ (approach 2) or no smoothing (approach 3), we examine the 1-dimensional cases where these methods perform worst relative to using $h_Y$ (approach 1). This shows some of the qualitative impacts of suboptimal smoothing on the density estimate.

We first study $\sigma_{\epsilon}^2=2$, Normal. Here $h_X$ performed worst (relative to $h_Y$) out of all the densities and error variances considered (see Table \ref{tab:50}). We generate 100 samples of size 50 from Normal. Using these 100 samples, we construct 100 density estimates using $h_Y$ and $h_X$. In Figure \ref{fig:normal_confs_draws} a) we plot the .1 and .9 pointwise quantiles for these density estimates (orange--dashed for $h_Y$ and blue--dotted for $h_X$) along with the true underlying density $f_Y$ (i.e. the Normal density convolved with $\phi_2$) in black--solid. The quantiles for $h_X$ have a lower peak and heavier tails than the quantiles for the $h_Y$ density estimates. Using $h_X$ oversmooths the density estimates ($h_X = 0.52$ and $h_Y = 0.26$).

In Figures \ref{fig:normal_confs_draws} b) and c) we plot 10 density estimates using $h_Y$ and $h_X$ respectively. We see that the individual density estimates using $h_X$ are negatively biased near $Y=0$ and positively biased for large $|Y|$. Since all the density estimates are unimodal, with a mode near 0 and approximately normal, the qualitative conclusions that one is likely to draw from these density estimates are likely to be similar regardless of whether one is using $h_X$ or $h_Y$.

\begin{figure}[h]
\begin{center}
\begin{tabular}{ccc}
$\begin{array}{ccc}
\multicolumn{1}{l}{\mbox{(a)}} & \multicolumn{1}{l}{\mbox{(b)}} & \multicolumn{1}{l}{\mbox{(c)}} \\ \\ \\ [-.5in]
\includegraphics[width=0.3\textwidth]{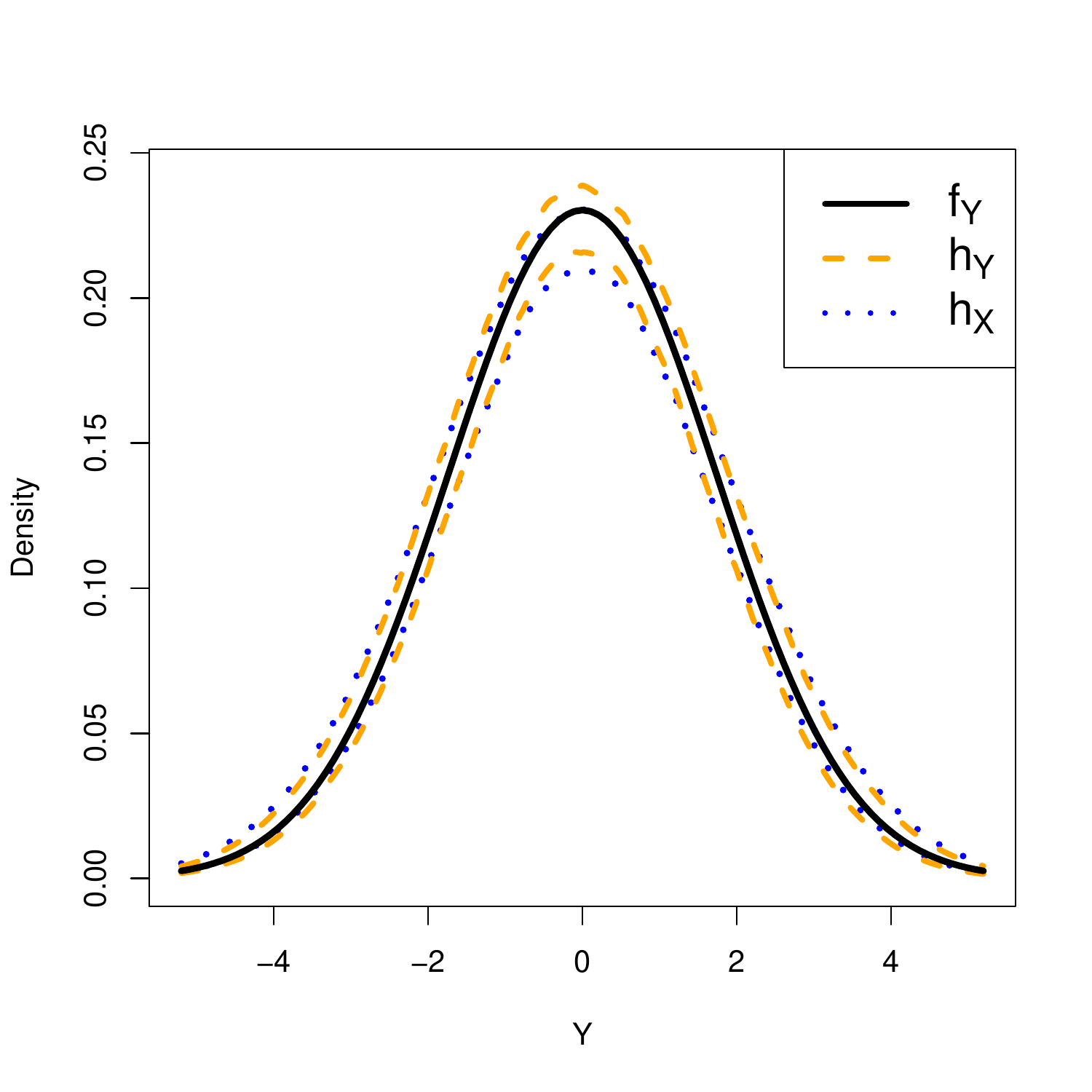} &
\includegraphics[width=0.3\textwidth]{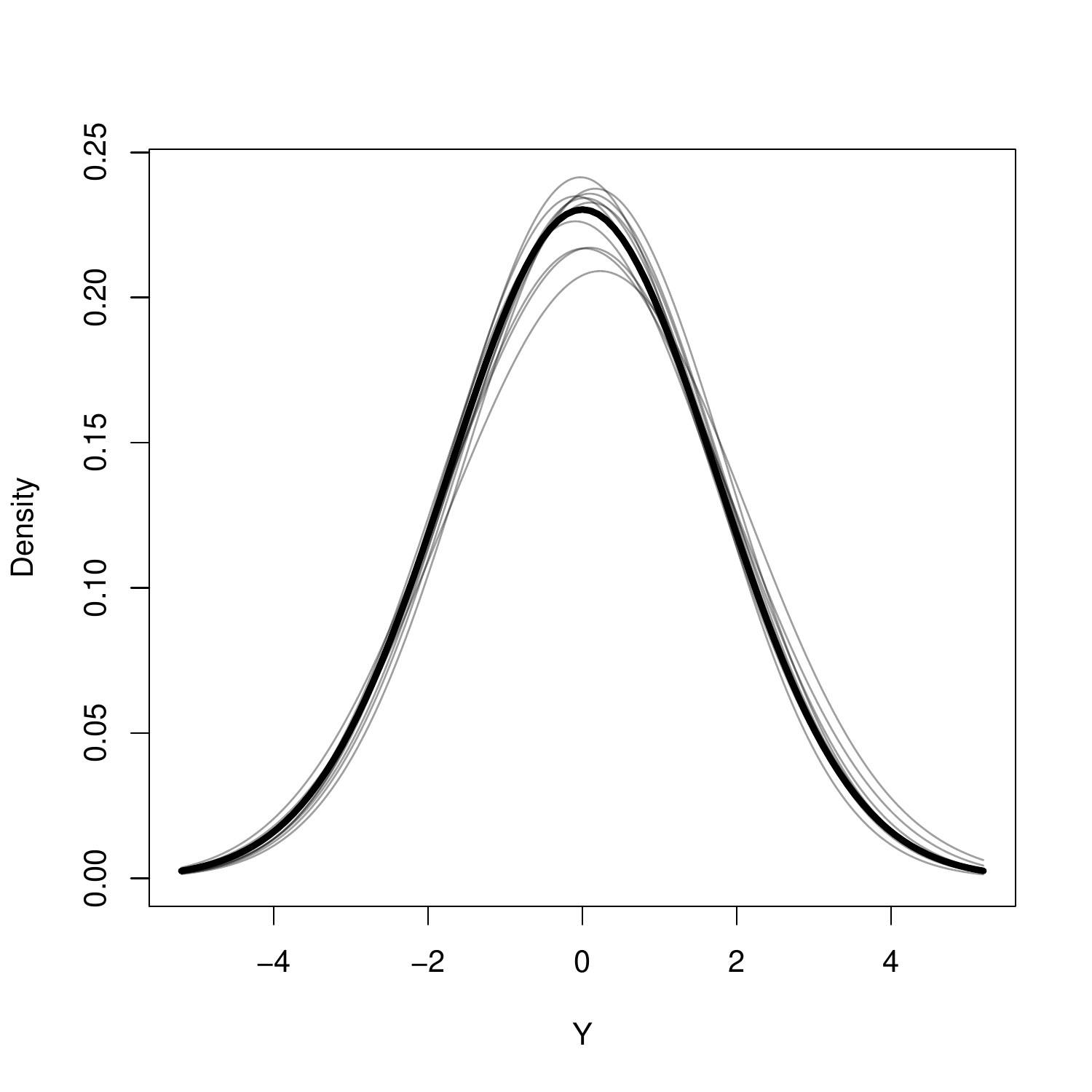} &
\includegraphics[width=0.3\textwidth]{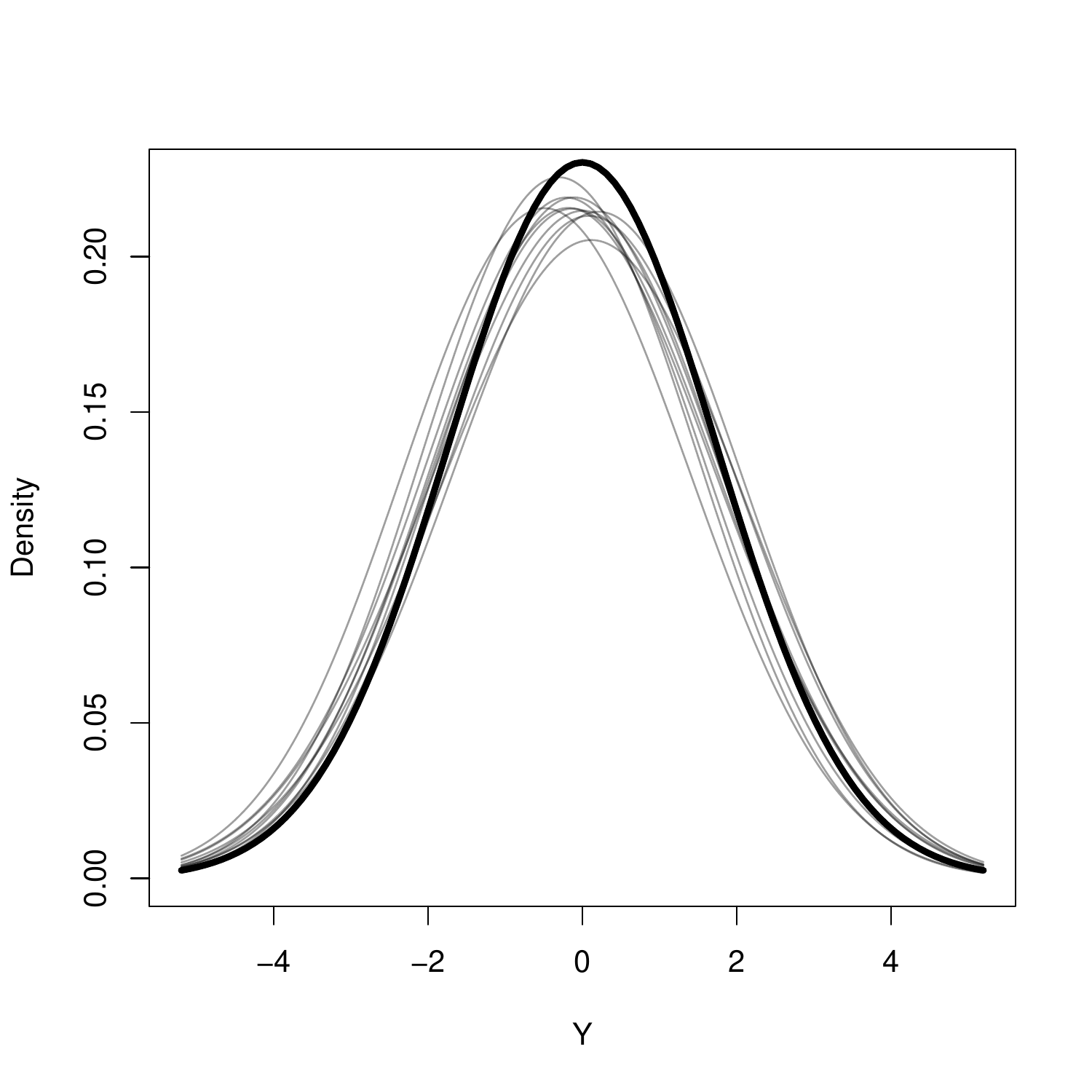}
\end{array}$
\end{tabular}
\end{center}
\vspace{-.3in}
\caption{Comparison of using $h_Y$ (optimal smoothing for $f_Y$) to $h_X$ (optimal smoothing for $f_X$) for the Normal density with $\sigma_{\epsilon}^2=2$. In a) we plot $f_Y$ and the .9 and .1 quantiles for density estimates using $h_Y$ (orange--dash) and $h_X$ (blue--dot). $h_X$ oversmooths the estimate, so the peak at $Y=0$ is biased low while the tails are biased high. In b) and c) we plot 10 density estimates using $h_Y$ and $h_X$ respectively. The qualitative conclusions that one is likely to draw from the density estimates are similar, regardless of whether $h_X$ or $h_Y$ is used.}
\label{fig:normal_confs_draws}
\end{figure}

We now study the Bimodal 1 density case with $\sigma_{\epsilon}^2 = .125$ and $n=50$. Here $\MISE(0)/MISE(h_Y)$ was highest out of all conditions tested in Table \ref{tab:50}. Following the procedure for generating Figure \ref{fig:normal_confs_draws}, we generate 100 samples of size $50$ from Bimodal 1. Using these 100 samples, we construct 100 density estimates using $h_Y$ and no smoothing. In Figure \ref{fig:bimodal1_confs_draws} a) we plot the .1 and .9 pointwise quantiles for these density estimates (orange--dashed for $h_Y$ and blue--dotted for no smoothing). We plot the true underlying density, $f_Y$ (i.e. the Bimodal 1 density convolved with $\phi_{.125}$), in black--solid.

No smoothing greatly overestimates the height of the mode at $Y=0$. The quantiles for the $h_Y$ density estimates are nearly contained within the quantiles for no smoothing across all values of $Y$. In Figure \ref{fig:bimodal1_confs_draws} b) and c) we plot 10 density estimates using $h_Y$ and no smoothing respectively. The density estimates using $h_Y$ (in b)) typically identify two modes close to the correct $Y$ values. In contrast the density estimates using no smoothing appear very undersmoothed (in c)). Several estimates have three or more modes and the mode heights are often far from the true value. In this case, not smoothing could have a significant impact on qualitative conclusions drawn from the density estimate. The results from Figures \ref{fig:normal_confs_draws} and \ref{fig:bimodal1_confs_draws} suggest that the qualitative impacts of not smoothing may be worse than smoothing using $h_X$.

\begin{figure}[h]
\begin{center}
\begin{tabular}{ccc}
$\begin{array}{ccc}
\multicolumn{1}{l}{\mbox{(a)}} & \multicolumn{1}{l}{\mbox{(b)}} & \multicolumn{1}{l}{\mbox{(c)}} \\ \\ \\ [-.5in]
\includegraphics[width=0.3\textwidth]{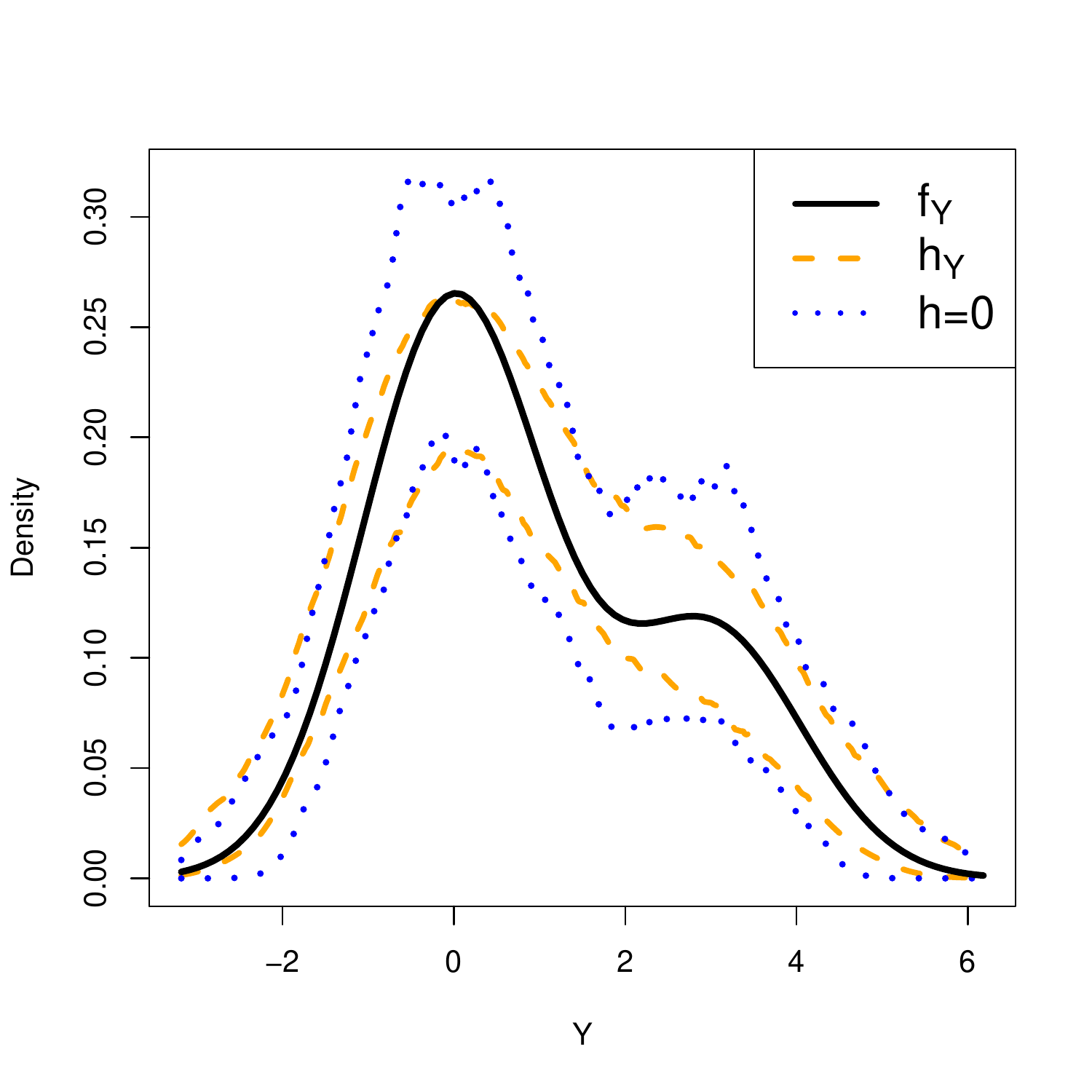} &
\includegraphics[width=0.3\textwidth]{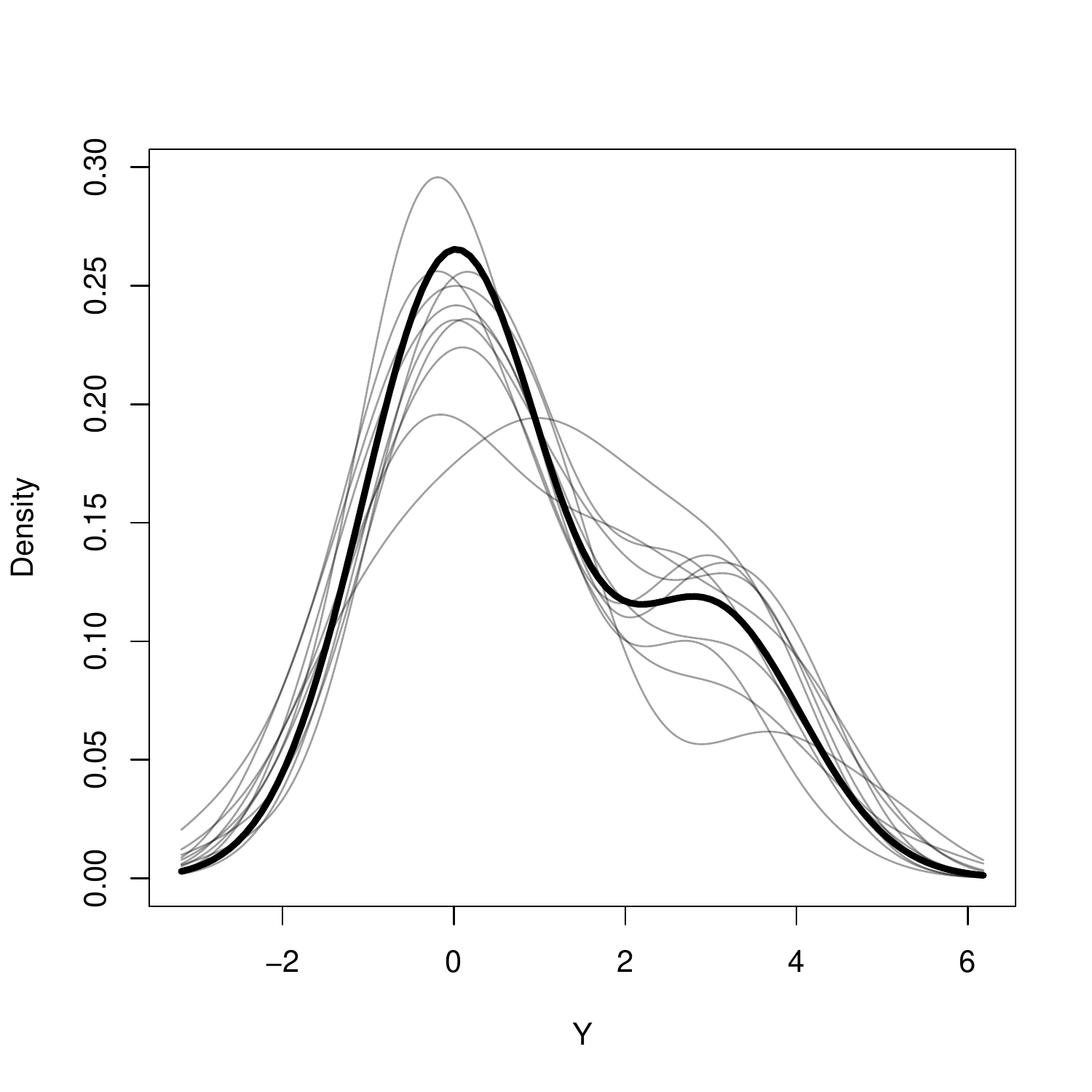} &
\includegraphics[width=0.3\textwidth]{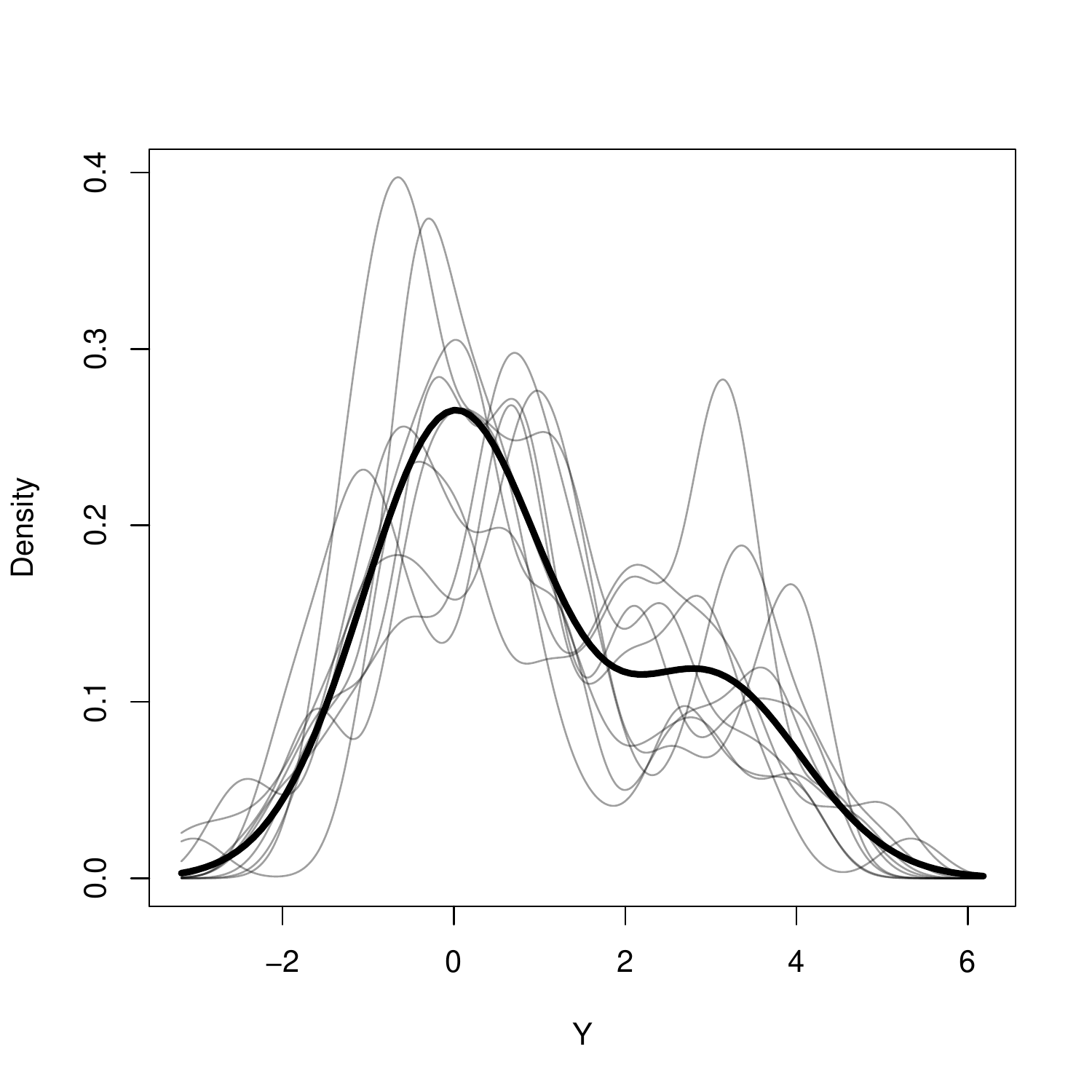}
\end{array}$
\end{tabular}
\end{center}
\vspace{-.3in}
\caption{Comparison of using $h_Y$ (optimal smoothing for $f_Y$) to no smoothing for the Bimodal 1 density with $\sigma_{\epsilon}^2=.125$. In a) we plot $f_Y$ (black--solid) and the .9 and .1 quantiles for density estimates using $h_Y$ (orange--dash) and no smoothing (blue--dot). The quantiles for no smoothing are wider than for $h_Y$ for most values of $Y$. In b) and c) we plot 10 density estimates using $h_Y$ and no smoothing respectively. The density estimates using no smoothing often have 3 modes. These modes are often not close to the true $Y$ value modes.}
\label{fig:bimodal1_confs_draws}
\end{figure}

\subsection{Convergence of Bandwidth to Asymptotic Approximation}
\label{sec:conv_rate}

We now study the rate of convergence of the asymptotically optimal bandwidth for estimating $f_Y$, $h_Y^*$ (see Equation \ref{eq:1dasymptopt}) to the exact optimal bandwidth $h_Y$. Fast convergence rates suggest that plug--in estimators could be effective for estimating $h_Y$. We pay particular attention to the relationship between convergence rate of $h_Y^*$ and the error variance $\sigma_{\epsilon}^2$.

For the four densities in Figure \ref{fig:densities} using error variances $\sigma_{\epsilon}^2 = 2,1,.5,.25,.125$ we compute the ratio between the exact optimal bandwidth ($h_Y$) and the asymptotically optimal bandwidth ($h_Y^*$). The exact optimal bandwidth is determined by finding the $h$ which minimizes Equation \ref{eq:exact_mise}. The asymptotically optimal bandwidth is computed using Equation \ref{eq:1dasymptopt}. We plot these ratios as a function of $n$ for the Normal and Trimodal densities in Figure \ref{fig:h_asymptotics} a) and b) respectively.

For the Normal density, the larger $\sigma_{\epsilon}^2$, the faster the convergence of the asymptotically optimal bandwidth to the actual optimal bandwidth. For example with $\sigma_{\epsilon}^2 = 2,1,.5$ and $n=100$, the exact optimal bandwidth is within 10\% of the asymptotic expression. This suggests that for a normal density, with moderate $n$ and error variance not too small, plug--in estimators for the asymptotic bandwidth may provide a good approximation to the bandwidth which minimizes the exact $\MISE$. The plots for Bimodal 1 and Bimodal 2 (not shown) closely resemble the Normal density.

For the Trimodal density, the relationship between $\sigma_{\epsilon}^2$ and the convergence rate of the asymptotically optimal bandwidth ($h_Y^*$) to exact optimal bandwidth ($h_Y$) is more complicated than for the Normal, Bimodal 1, and Bimodal 2 densities. Broadly, the asymptotics for $\sigma_{\epsilon}^2 = .25,.125$ take hold at larger $n$ than for $\sigma_{\epsilon}^2 = 2,1,.5$. Unlike for the normal case, the asymptotically optimal bandwidth is not always greater than the exact optimal bandwidth.

\begin{figure}[h]
\begin{center}
\begin{tabular}{cc}
$\begin{array}{cc}
\multicolumn{1}{l}{\mbox{(a)}} & \multicolumn{1}{l}{\mbox{(b)}} \\ \\ \\ [-.7in]
\includegraphics[width=0.45\textwidth]{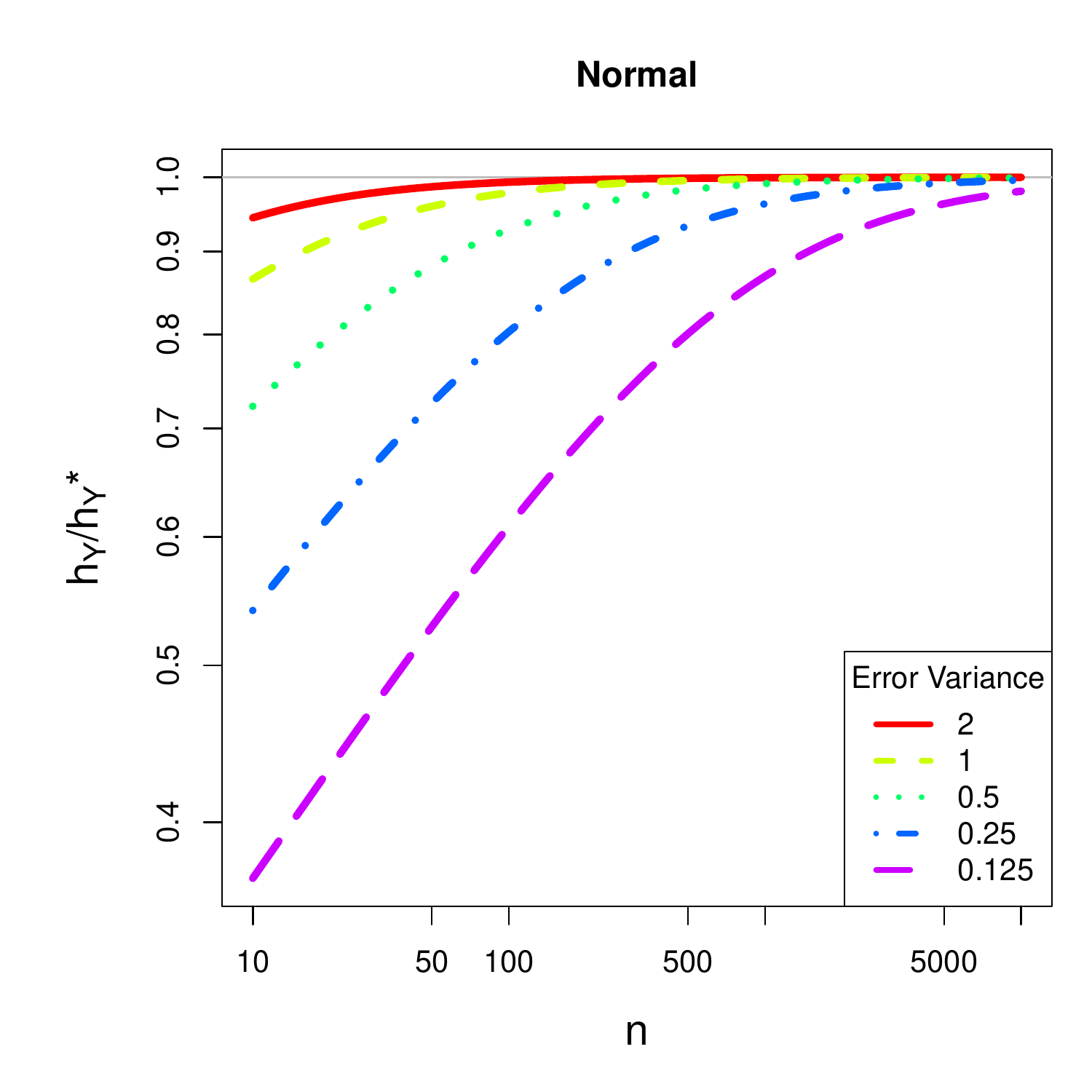} &
\includegraphics[width=0.45\textwidth]{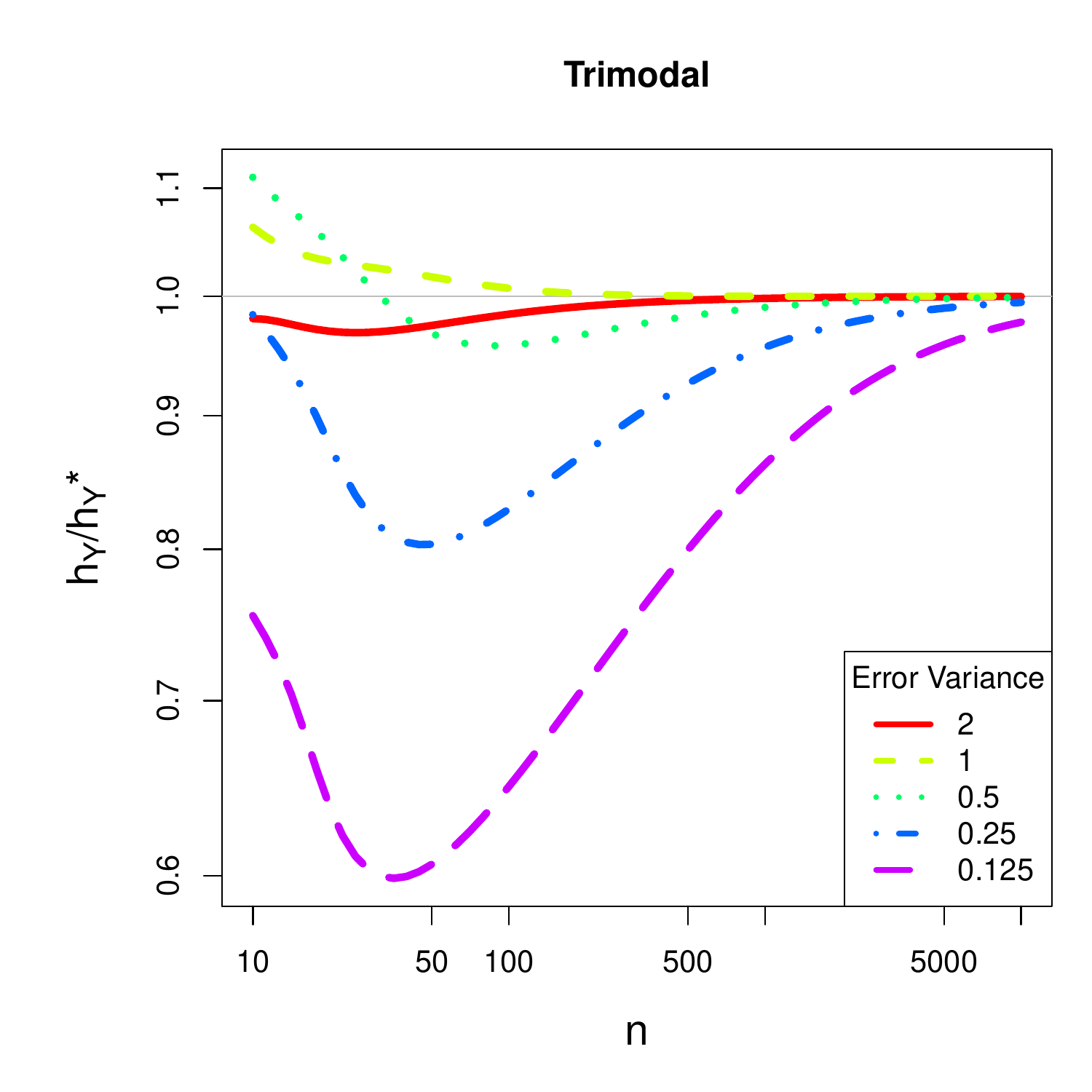}
\end{array}$
\end{tabular}
\end{center}
\caption{Ratio of exact optimal bandwidth to asymptotically optimal bandwidth ($h_Y / h_Y^*$) as a function of $n$ for the Normal (a) and Trimodal (b) densities. The convergence of this ratio to $1$ varies with $\sigma_{\epsilon}^2$ (the error variance). For the Normal density the relationship between $\sigma_{\epsilon}^2$ and the convergence is fairly simple, while for the Trimodal, the behavior is more complex.}
\label{fig:h_asymptotics}
\end{figure}

Certain aspects of the convergence rate behavior in Figure \ref{fig:h_asymptotics} may be explained by considering asymptotics in $\sigma_{\epsilon}^2$. At constant $n$, as $\sigma_{\epsilon}^2 \rightarrow 0$, $h_Y / h_Y^* \rightarrow 0$. This can be seen by considering the limiting values (in $\sigma_{\epsilon}^2$) of $h_Y$ and $h_Y^*$. Note that as $\sigma_{\epsilon}^2 \rightarrow 0$, $f_Y \rightarrow f_X$. Therefore $h_Y \rightarrow h_X$, a positive constant. In contrast, as $\sigma_{\epsilon}^2 \rightarrow 0$, $h_Y^* = \sqrt{\frac{2\int (\omega^T \Sigma_K \omega) d\nu(\omega)}{n\int (\omega^T \Sigma_K \omega)^2 d\mu(\omega)}} \rightarrow \infty$. This behavior is seen in Figure \ref{fig:h_asymptotics} a) where at fixed $n$, the smaller $\sigma_{\epsilon}^2$, the smaller $h_Y/h_Y^*$. For $n > 1000$, this relationship between $\sigma_{\epsilon}^2$ and $h_Y / h_Y^*$ is true for Trimodal (Figure \ref{fig:h_asymptotics} b)) as well. Consideration of asymptotic regimes in which $n \rightarrow \infty$ and $\sigma_{\epsilon}^2 \rightarrow 0$ together at some rate may help explain the behavior in Figure \ref{fig:h_asymptotics} better.

In cases where the asymptotically optimal bandwidth is far from the exact optimal bandwidth, estimators for the asymptotically optimal bandwidth (such as plug--in or ``rule--of--thumb'' estimators) may have poor performance in terms of minimizing risk. The simulation results in Figure \ref{fig:h_asymptotics} suggest that when $\epsilon$ is concentrated around $0$, $h_Y^*$ is far from $h_Y$. Thus plug--in estimators, which attempt to estimate $h_Y^*$), may be suboptimal in terms of minimizing $\MISE$. In Section \ref{sec:real_data} we observe this behavior with a ``rule--of--thumb'' estimator for $h_Y$ when $\epsilon$ has small variance.

\section{Estimator for $h_Y$ and Real Data Example}
\label{sec:real_data}

\subsection{Rule--of--Thumb Estimator for $h_Y$}
\cite{jones1996brief} define ``rule--of--thumb'' bandwidth selection procedures as any method which replaces unknown quantities in the asymptotically optimal bandwidth with estimated values based on a parametric family for the unknown density. We now propose a rule--of--thumb estimation method for $h_Y$. Recall from Equation \eqref{eq:1dasymptopt} that the asymptotically optimal bandwidth is
\begin{equation}
\label{eq:h_star_norm}
h^*_Y = \sqrt{\frac{2\int (\omega^T \Sigma_K \omega) d\nu(\omega)}{n\int (\omega^T \Sigma_K \omega)^2 d\mu(\omega)}} = \sqrt{\frac{2\int (\omega^T \Sigma_K \omega) |\widehat{f}_\epsilon(\omega)|^2(1 - |\widehat{f}_X(\omega)|^2)d\omega}{n\int (\omega^T \Sigma_K \omega)^2 |\widehat{f}_X(\omega)|^2|\widehat{f}_\epsilon(\omega)|^2 d\omega}}.
\end{equation}
We specialize to one dimension, so $\Sigma_K$ is a scalar and can be set to 1 without loss of generality. $h_Y^*$ depends on $|\widehat{f}_X(\omega)|^2$, which is unknown. We replace this quantity by assuming (solely for the purposes of bandwidth estimation) that $f_X$ is mean $0$ normal. In this case, $|\widehat{f}_X(\omega)|^2 = e^{-\sigma_X^2\omega^2}$ where $\sigma_X^2$ is the variance of $f_X$. $\sigma_X^2$ is estimated with $\widetilde{\sigma}^2_X$, the variance of the observations $X_1, \ldots, X_n$. Thus our rule--of--thumb bandwidth estimator for Berkson kernel density estimation is
\begin{equation}
  \label{eq:hy_estimate_gen}
\widetilde{h}_Y = \sqrt{\frac{2\int \omega^2 |\widehat{f}_\epsilon(\omega)|^2(1 - e^{-\widetilde{\sigma}_X^2\omega^2})d\omega}{n\int \omega^4 e^{-\widetilde{\sigma}_X^2\omega^2}|\widehat{f}_\epsilon(\omega)|^2 d\omega}}.
\end{equation}
For the case where $\epsilon$ is mean $0$ normal with variance $\sigma_{\epsilon}^2$, $|\widehat{f}_\epsilon(\omega)|^2 = e^{-\sigma_{\epsilon}^2\omega^2}$ and Equation \ref{eq:hy_estimate_gen} simplifies to
\begin{equation}
  \label{eq:hy_estimate}
\widetilde{h}_Y = \sqrt{\frac{4}{3n}\left[ \frac{(\widetilde{\sigma}_X^2 + \sigma_\epsilon^2)^{5/2}}{\sigma_\epsilon^2} - (\sigma_{\epsilon}^2 + \widetilde{\sigma}_X^2)\right]}.
\end{equation}

\subsection{Real Data Example}

We analyze data collected by \cite{ferris1979effects} concerning childhood exposure to \no, a known cause of respiratory illnesss. The goal is to determine the density of exposure to \no for children living in Watertown, Massachusetts. \cite{ferris1979effects} collected kitchen and bathroom concentrations of \no for 231 homes in Watertown. In this study, direct personal exposure to \no was not observed.

Using data collected in Portage, Wisconsin and the Netherlands, \cite{tosteson1989measurement} modeled log personal exposure to \no ($Y$) as a linear function of log kitchen ($\ln(W_k)$) and log bathroom ($\ln(W_b)$) concentrations plus random error. Specifically (see Table 1 of \cite{tosteson1989measurement})
\begin{equation*}
Y = 1.22 + 0.3\ln(W_k) + 0.33 \ln(W_b) + \epsilon
\end{equation*}
where $\epsilon \sim N(0,.06)$, independent of $W_k$ and $W_b$. Let $X = 1.22 + 0.3\ln(W_k) + 0.33 \ln(W_b)$. By assuming the same error model holds in Watertown as in Portage and the Netherlands (i.e. assuming portability of the error model), we can estimate log personal exposure density in Watertown as $X$ plus independent noise where we have 231 observed values of $X$.

We estimate the density of $Y$ using the three smoothing approaches described in earlier sections: smoothing to optimize estimation of $f_Y$, smoothing to optimize estimation of $f_X$ and no smoothing. For no smoothing we simply convolve the observations with the error density. For smoothing to optimize estimation of $f_X$ and $f_Y$ we must select a kernel $K$ and bandwidth estimation methods for $h_X$ and $h_Y$. We use a Gaussian kernel. For estimating $h_X$ we use ``Silverman's rule--of--thumb'', developed in \cite{deheuvels1977estimation} and \cite{silverman1986density},\footnote{This is the default bandwidth selection for the \begin{tt}density\end{tt} function in the software package \begin{tt}R\end{tt}, Version 3.01.}
\begin{equation*}
\widetilde{h}_X = \frac{0.9\min(\widetilde{IQR}_X/1.34,\widetilde{\sigma}_X)}{n^{1/5}}.
\end{equation*}
Here $\widetilde{IQR}$ and $\widetilde{\sigma}_X$ are the estimated inter--quartile range and standard deviation of $X$. For estimating $h_Y$ we use the rule--of--thumb estimator $\widetilde{h}_Y$ proposed in Equation \eqref{eq:hy_estimate}.

In addition to studying the case $\epsilon \sim N(0,.06)$, we construct density estimates when $\epsilon$ is normal with variance $0.6$ and $0.006$ in order to study robustness of the smoothing methods to different levels of Berkson error. In Figure \ref{fig:real_data} we plot the three estimators for a) $\epsilon \sim N(0,0.6)$, b) $\epsilon \sim N(0,0.06)$, and c) $\epsilon \sim N(0,0.006)$. In a) where $\sigma_{\epsilon}^2 = 0.6$ there is essentially no difference in the estimators. In b) where $\sigma^2_\epsilon = 0.06$ no smoothing results in a somewhat higher mode around $y=3$ than smoothing to optimize estimation of $f_X$ or $f_Y$. It appears unlikely that the choice of smoothing would affect qualitative conclusions in this case. In c) where $\sigma^2_\epsilon = 0.006$ no smoothing severely under--regularizes the density estimate. In particular the estimate has four modes. $\widetilde{h}_Y$ oversmooths the estimate of $f_Y$. This is likely due to the fact that for fixed $n$ as $\sigma_{\epsilon} \rightarrow 0$, $h^*_Y \rightarrow \infty$ while $h_Y \rightarrow h_X$ (see Subsection \ref{sec:conv_rate}). Since $\widetilde{h}_Y$ is an estimate of $h^*_Y$, it oversmooths the density estimate in this case. Overall, with $\epsilon \sim N(0,0.006)$, $\widetilde{h}_X$ produces the best density estimate (blue dashed line).

\begin{figure}[h]
\begin{center}
\begin{tabular}{ccc}
$\begin{array}{ccc}
\multicolumn{1}{l}{\mbox{(a)}} & \multicolumn{1}{l}{\mbox{(b)}} & \multicolumn{1}{l}{\mbox{(c)}} \\ \\ \\ [-.5in]
\includegraphics[width=0.3\textwidth]{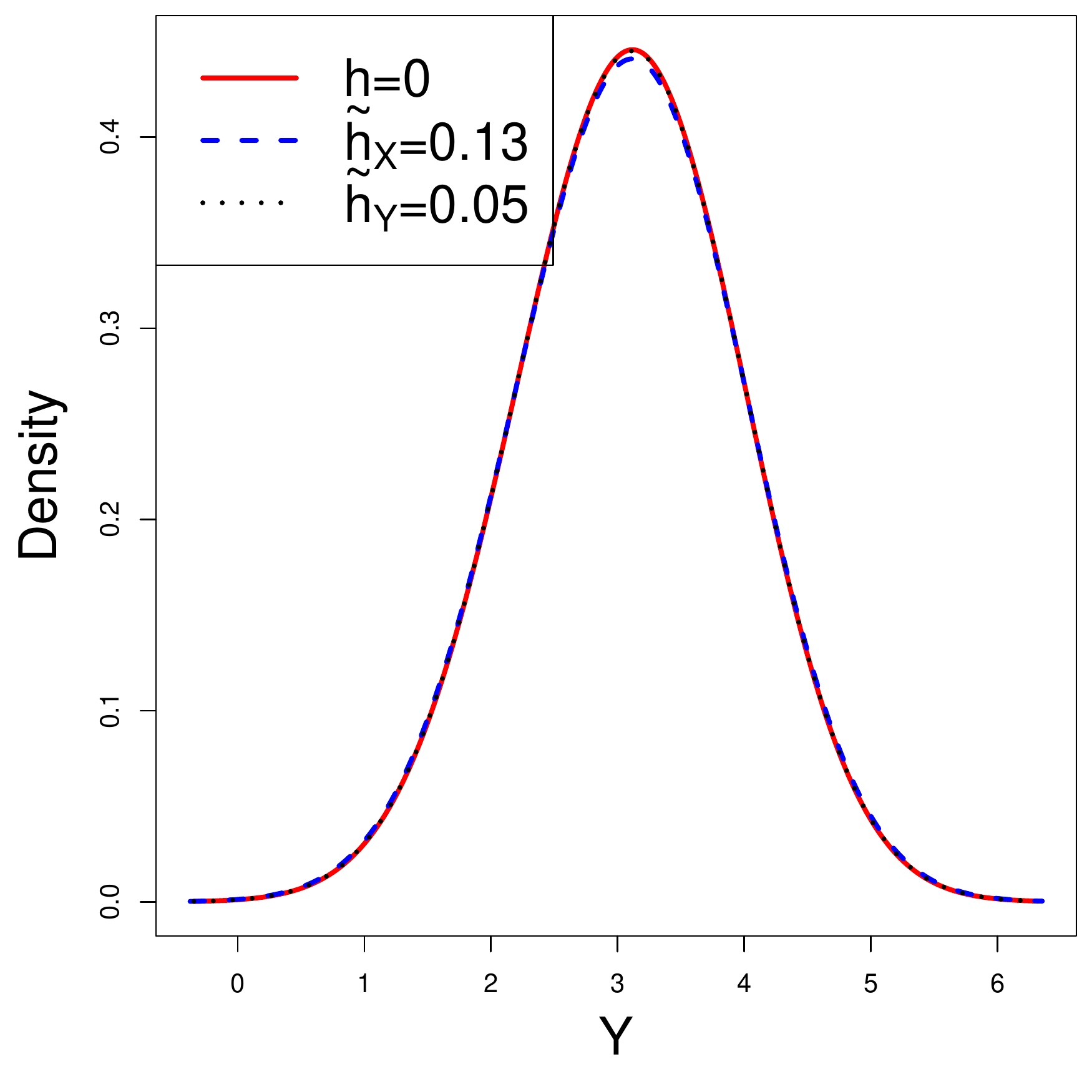} &
\includegraphics[width=0.3\textwidth]{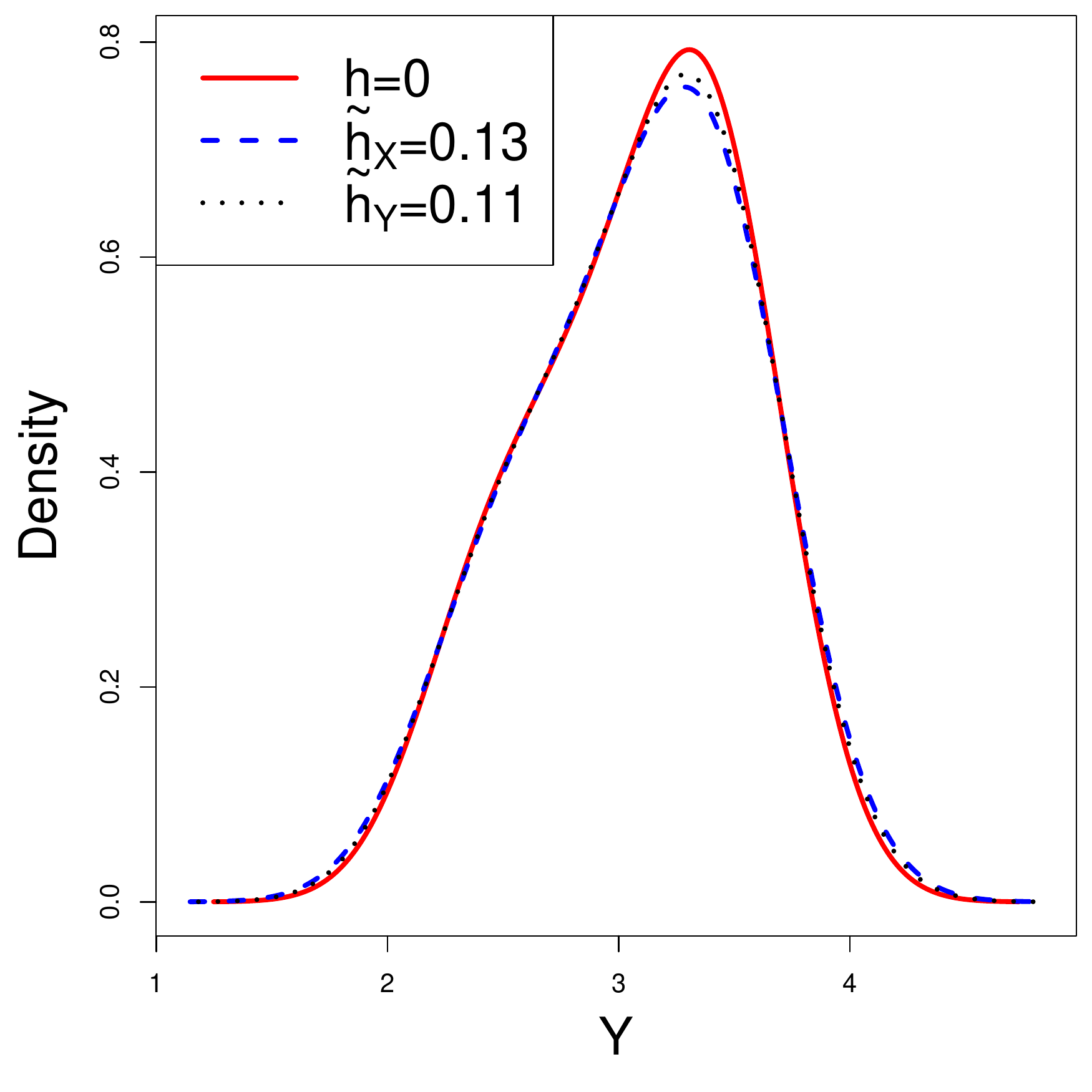} &
\includegraphics[width=0.3\textwidth]{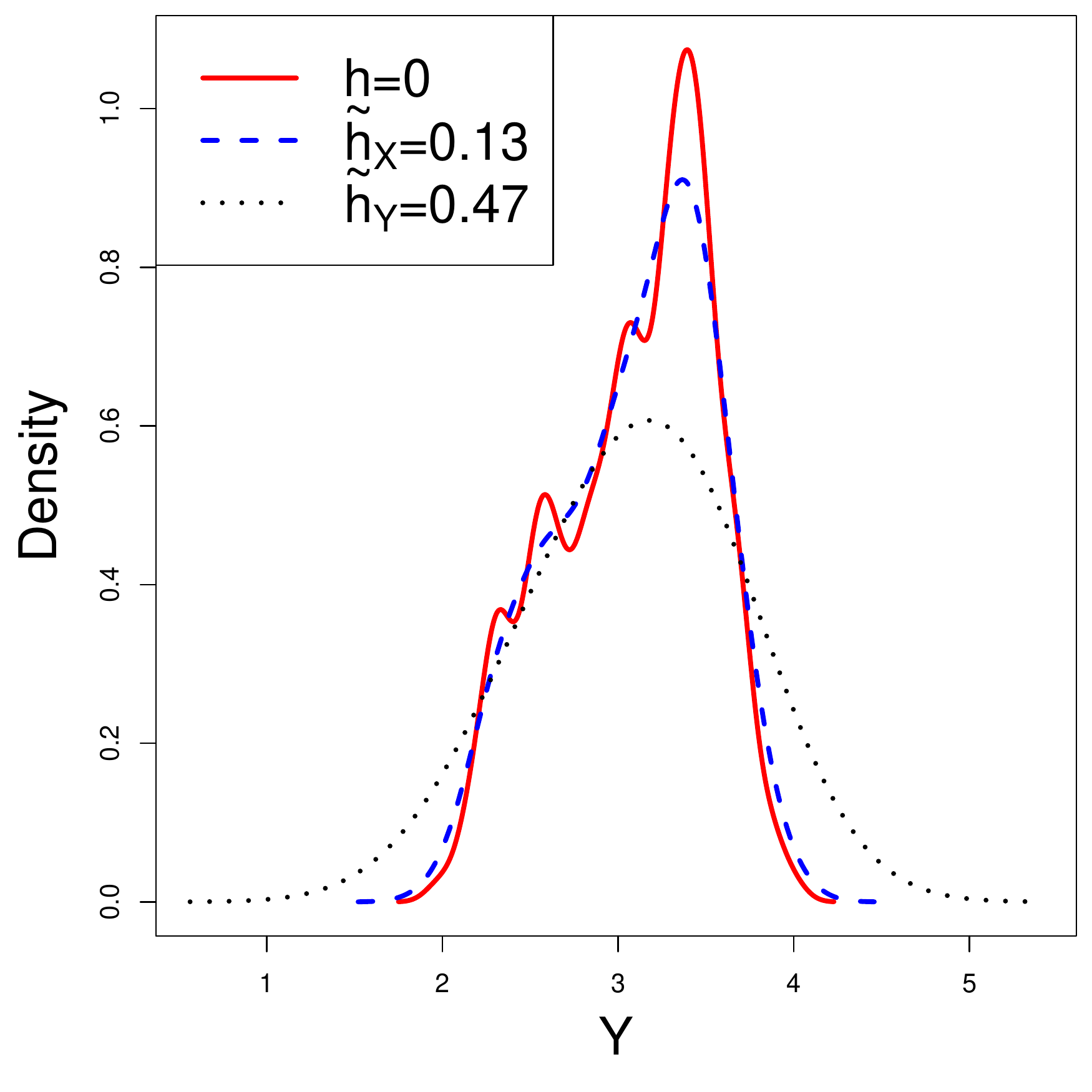}
\end{array}$
\end{tabular}
\end{center}
\vspace{-.3in}
\caption{Density estimates of log \no exposure for children living in Watertown using three different error variances. With error variances $\sigma^2_{\epsilon} = 0.6$ or $0.06$, plots a) and b) respectively, all three smoothing methods produce similar density estimates. In c), where $\sigma_{\epsilon}^2 = 0.006$, no smoothing under--regularizes the density estimate.}
\label{fig:real_data}
\end{figure}

\section{Conclusions}
\label{sec:conclusions}

In this work we compared different approaches to smoothing a density estimate subject to Berkson error. No smoothing (approach 3) achieved suboptimal asymptotic (at second order) and finite sample $\MISE$. This was especially evident when the error term $\epsilon$ was concentrated near 0. Smoothing to optimize estimation of $f_X$ resulted in suboptimal asymptotic (at first order) $\MISE$ rates. At finite samples, $h_X$ oversmoothed density estimates, particularly when the error variance was large. These effects grew worse in higher dimensions.

These results support using a bandwidth specifically chosen for estimation of $f_Y$. More work is needed to develop estimators of $h_Y$. The asymptotically optimal bandwidth $h_Y^*$ derived in Equation \eqref{eq:1dasymptopt} suggest one form for rule--of--thumb and plug--in type estimators. The simulations in Subsection \ref{sec:conv_rate} suggest that when the error is concentrated around $0$, using a bandwidth which estimates $h_Y^*$ (the asymptotic approximation to $h_Y$) may oversmooth the density estimate. The rule--of--thumb estimator for $h_Y$ we developed in Equation \eqref{eq:hy_estimate} displayed this behavior. Estimators for $h_Y$ based on cross--validation or the bootstrap may perform better under a wider range of possible error distributions and should be developed. The development of bandwidth estimators for standard kernel density estimation (see \cite{jones1996brief} for a review) suggest forms for such procedures.

\baselineskip=14pt
\section*{Acknowledgments}

Support from NSF grant 0941742 (Cyber-Enabled Discovery and Innovation), NSF grant DMS-0847647 (CAREER), and a fellowship from Citadel LLC are gratefully acknowledged. We thank Raymond Carroll for providing helpful comments on the manuscript and Len Stefanski for supplying the data set on children's \no exposure.

\baselineskip=14pt

\bibliographystyle{abbrvnat}
\bibliography{refs}

\clearpage\pagebreak\newpage
\pagestyle{fancy}
\fancyhf{}
\rhead{\bfseries\thepage}
\lhead{\bfseries SUPPLEMENTARY MATERIAL}
\begin{center}
{\LARGE{\bf Supplementary Material to\\ {\it Kernel Density Estimation with Berkson Error}}}
\end{center}

\vskip 2mm
\begin{center}
James P. Long\\
Department of Statistics, Texas A\&M University\\
3143 TAMU, College Station, TX 77843-3143\\
jlong@stat.tamu.edu\\
\hskip 5mm \\
Noureddine El Karoui \\
Department of Statistics, University of California, Berkeley\\
367 Evans Hall \# 3860, Berkeley, CA 94720-3860\\
nkaroui@stat.berkeley.edu\\
\hskip 5mm \\
John A. Rice \\
Department of Statistics, University of California, Berkeley\\
367 Evans Hall \# 3860, Berkeley, CA 94720-3860\\
rice@stat.berkeley.edu
\end{center}

\setcounter{equation}{0}
\setcounter{page}{1}
\setcounter{table}{1}
\setcounter{section}{0}
\renewcommand{\theequation}{S.\arabic{equation}}
\renewcommand{\thesection}{S.\arabic{section}}
\renewcommand{\thesubsection}{S.\arabic{section}.\arabic{subsection}}
\renewcommand{\thepage}{S.\arabic{page}}
\renewcommand{\thetable}{S.\arabic{table}}
\baselineskip=17pt

\section{Proofs of the Theorems}
\label{sec:proofs}
\subsection{Proof of Theorem \ref{thm:char_mise}}
\label{sec:char_mise}
We must show
\begin{equation*}
(2\pi)^p \MISE(H) = \int |1 - \widehat{K}(H\omega)|^2 d\mu(\omega) + \frac{1}{n}\int |\widehat{K}(H\omega)|^2 d\nu(\omega)
\end{equation*}
where
\begin{align*}
&d\mu(\omega) = |\widehat{f}_{\epsilon}(\omega)|^2|\widehat{f}_X(\omega)|^2 d\omega,\\
&d\nu(\omega) = |\widehat{f}_{\epsilon}(\omega)|^2 (1-|\widehat{f}_X(\omega)|^2) d\omega.
\end{align*}
Substituting for $d\mu(\omega)$ and $d\nu(\omega)$, it suffices to show that
\begin{equation}
\label{eq:suffices}
(2\pi)^p \MISE(H) = \int |\widehat{f}_{\epsilon}(\omega)|^2 \left(|1 - \widehat{K}(H\omega)|^2|\widehat{f}_X(\omega)|^2 + \frac{1}{n}|\widehat{K}(H\omega)|^2 (1-|\widehat{f}_X(\omega)|^2)\right)d\omega.
\end{equation}
$\ffyh, \widehat{f}_Y \in L_1$ by assumption. They are also in $L_2$ because they are characteristic functions and thus bounded. Under these conditions, the Plancherel theorem (see Theorem 1.8.8 on page 57 in \cite{ushakov1999selected}) states
\begin{equation}
\label{eq:time_freq}
\int (f_Y(y) - \fy)^2 dy = \frac{1}{(2\pi)^p} \int |\widehat{f}_Y(\omega) - \ffy|^2 d\omega.
\end{equation}
Let $\mathcal{P}_n$ be the product measure on $(X_1, \ldots, X_n)$. Using the definition of $\MISE(H)$, Equation \eqref{eq:time_freq},  and the facts $\widehat{f}_Y(\omega) = \widehat{f}_X(\omega)\widehat{f}_{\epsilon}(\omega)$ and $\ffy = \widehat{K}(H\omega)\widehat{f}_{\epsilon}(\omega)\ffx$, we have
\begin{align*}
\MISE(H) &= \E_{\mathcal{P}_n}\int \left(f_Y(y) - \fy \right)^2 dy \\
&= \frac{1}{(2\pi)^p}\E_{\mathcal{P}_n}\int |\widehat{f}_Y(\omega) - \ffy|^2 d\omega\\
&= \frac{1}{(2\pi)^p}\E_{\mathcal{P}_n}\int |\widehat{K}(H\omega)\widehat{f}_{\epsilon}(\omega)\ffx - \widehat{f}_X(\omega)\widehat{f}_{\epsilon}(\omega)|^2 d\omega\\
&= \frac{1}{(2\pi)^p}\E_{\mathcal{P}_n}\int |\widehat{f}_{\epsilon}(\omega)|^2|\ffx\widehat{K}(H\omega) - \widehat{f}_X(\omega)|^2 d\omega.
\end{align*}
Note that the integrand is a non-negative function, so we move the expectation inside the integral using Fubini's Theorem. We have
\begin{equation*}
(2\pi)^p \MISE(H) = \int |\widehat{f}_{\epsilon}(\omega)|^2\E_{\mathcal{P}_n}|\ffx\widehat{K}(H\omega) - \widehat{f}_X(\omega)|^2  d\omega.
\end{equation*}
Noting that it is sufficient to show Equation \eqref{eq:suffices} holds, all that is left is to show is
\begin{equation*}
\E_{\mathcal{P}_n}|\ffx\widehat{K}(H\omega) - \widehat{f}_X(\omega)|^2 = |1 - \widehat{K}(H\omega)|^2|\widehat{f}_X(\omega)|^2 + \frac{1}{n}|\widehat{K}(H\omega)|^2 (1-|\widehat{f}_X(\omega)|^2).
\end{equation*}
This identity is shown in the proof of Theorem 1.4 on page 22 in \cite{tsybakov2009introduction}.
\hfill $\qed$

\subsection{Proof of Theorem \ref{thm:multi_dim}}
\label{sec:multi_d}
Recall that we are working under Assumptions \ref{assump:kernel} and \ref{assump:density}. This proof is divided into three parts. In \textbf{Part 1} we show $\widehat{f}_Y, \ffyh \in L_1$, which satisfies the conditions for Theorem \ref{thm:char_mise} and implies
\begin{equation}
\label{eq:known_mise}
(2\pi)^p \MISE(H) = \int |1 - \widehat{K}(H\omega)|^2 d\mu(\omega) + \frac{1}{n}\int |\widehat{K}(H\omega)|^2 d\nu(\omega).
\end{equation}
In \textbf{Part 2} we expand the first term of the right hand side of Equation \eqref{eq:known_mise} to show
\begin{equation}
\label{eq:bias2}
\int |1 - \widehat{K}(H\omega)|^2 d\mu(\omega) = \left(\frac{1}{4} \int (\omega^TH^T\Sigma_KH\omega)^2d\mu(\omega)\right) (1 + O(||H||_{\infty}^2)).
\end{equation}
In \textbf{Part 3} we expand the second term of the right hand side of Equation \eqref{eq:known_mise} to show
\begin{equation}
\label{eq:var2}
\frac{1}{n}\int |\widehat{K}(H\omega)|^2 d\nu(\omega) = \frac{1}{n}\int d\nu(\omega) - \left(\frac{1}{n}\int (\omega^T H^T \Sigma_K H \omega) d\nu(\omega)\right)(1 + O(||H||_{\infty}^2)).
\end{equation}
Summing Equations \eqref{eq:bias2} and \eqref{eq:var2} we have the result
\begin{align*}
&(2\pi)^p \MISE(H) \\
&= \frac{1}{n}\int d\nu(\omega) + \\
&+\left(\frac{1}{4} \int (\omega^TH^T \Sigma_K H\omega)^2d\mu(\omega) - \frac{1}{n}\int (\omega^T H^T \Sigma_K H \omega) d\nu(\omega)\right)(1 + O(||H||_{\infty}^2)).
\end{align*}

\vspace{.1in}
\noindent
\textbf{Part 1: $\widehat{f}_Y, \ffyh \in L_1$}\\
Note that since the modulus of a characteristic function is bounded by 1
\begin{align*}
&|\widehat{f}_Y(\omega)| = |\widehat{f}_X(\omega)\widehat{f}_{\epsilon}(\omega)| \leq |\widehat{f}_{\epsilon}(\omega)|,\\
&|\ffy| = |\widehat{K}(H\omega)\widehat{f}_{\epsilon}(\omega)\ffx| \leq |\widehat{f}_{\epsilon}(\omega)|.
\end{align*}
$\widehat{f}_{\epsilon} \in L_1$ by Assumption \eqref{assump:ep_l1}, implying $\widehat{f}_Y,\ffyh \in L_1$.

\vspace{.1in}
\noindent
\textbf{Part 2: Bias}\\
By Lemma \ref{lemma:taylor_error_bound} on \page \pageref{lemma:taylor_error_bound} there exists $R$ satisfying
\begin{equation}
\label{eq:remainder_bound}
|R(\omega)| \leq C ||\omega||_{\infty}^4
\end{equation}
such that
\begin{equation}
\label{eq:khatexpansion}
  \widehat{K}(\omega) = 1 - \frac{\omega^T\Sigma_K \omega}{2} + R(\omega).
\end{equation}
Note that the kernel $K$ is symmetric so $\widehat{K}$ and $R$ are real valued functions.
\begin{align}
  \int |1 - \widehat{K}(H\omega)|^2 d\mu(\omega) =& \int \left|\frac{\omega^TH^T\Sigma_KH\omega}{2} - R(H\omega)\right|^2 d\mu(\omega) \nonumber \\
=&\frac{1}{4} \int (\omega^TH^T\Sigma_KH\omega)^2d\mu(\omega) \nonumber \\
&-\int R(H\omega) (\omega^TH^T\Sigma_KH\omega)d\mu(\omega) \label{eq:bias_small1}\\
&+ \int R(H\omega)^2d\mu(\omega). \label{eq:bias_small2}
\end{align}
We have split the integrals formally. We now show that Expressions \eqref{eq:bias_small1} and \eqref{eq:bias_small2} are $O(||H||_{\infty}^6)$ by bounding their integrands.  Using the bound $R(\omega) \leq C||\omega||_{\infty}^4$ (Equation \eqref{eq:remainder_bound}), for some $E$ we have
\begin{align*}
&|R(H\omega) (\omega^TH^T\Sigma_KH\omega)| \leq C||H\omega||_{\infty}^4 ||\omega^TH^T\Sigma_KH\omega||_{\infty} \leq E ||H||_{\infty}^6||\omega||_{\infty}^6,\\
&|R(H\omega)^2| \leq C^2||H\omega||_{\infty}^8 \leq E ||H||_{\infty}^8||\omega||_{\infty}^8.
\end{align*}
Using the definition of $d\mu(\omega)$ and the fact $\int ||\omega||_{\infty}^8 |\widehat{f}_{\epsilon}(\omega)|d\omega < \infty$ (Assumption \eqref{assump:1ep}) we have
\begin{equation*}
\int ||\omega||_{\infty}^8 d\mu(\omega) = \int ||\omega||_{\infty}^8 |\widehat{f}_X(\omega)|^2|\widehat{f}_{\epsilon}(\omega)|^2 d\omega \leq \int ||\omega||_{\infty}^8 |\widehat{f}_{\epsilon}(\omega)|^2d\omega < \infty.
\end{equation*}
So Expressions \eqref{eq:bias_small1} and \eqref{eq:bias_small2} are $O(||H||^6)$ and $O(||H||^8)$ respectively. Thus
\begin{equation*}
  \int |1 - \widehat{K}(H\omega)|^2 d\mu(\omega) = \left(\frac{1}{4} \int (\omega^TH^T\Sigma_KH\omega)^2d\mu(\omega)\right) (1 + O(||H||_{\infty}^2)).
\end{equation*}

\vspace{.1in}
\noindent
\textbf{Part 3: Variance}
Using the expansion of $\widehat{K}$ in Equation \eqref{eq:khatexpansion} we have
\begin{equation*}
\frac{1}{n} \int |\widehat{K}(H\omega)|^2d\nu(\omega) = \frac{1}{n}\int \left|1 - \frac{\omega^T H^T \Sigma_K H \omega}{2} + R(H\omega)\right|^2 d\nu(\omega).
\end{equation*}
Expanding the right hand side we have
\begin{align}
\frac{1}{n}\int \left|1 - \frac{\omega^T H^T \Sigma_K H \omega}{2} + R(H\omega)\right|^2 d\nu(\omega)=&\frac{1}{n}\Big(\int d\nu(\omega) \\
&- \int (\omega^T H^T \Sigma_K H \omega) d\nu(\omega)\\
&+ \frac{1}{4}\int (\omega^T H^T \Sigma_K H \omega)^2 d\nu(\omega) \label{eq:vsmall1}\\
&- \int R(H\omega)(\omega^T H^T \Sigma_K H \omega) d\nu(\omega) \label{eq:vsmall2}\\
&+ 2\int R(H\omega) d\nu(\omega)\label{eq:vsmall3}\\
&+ \int R^2(H\omega) d\nu(\omega)\Big).\label{eq:vsmall4}
\end{align}
We have split the integral formally. Using the bound $R(\omega) \leq C||\omega||_{\infty}^4$ (Equation \eqref{eq:remainder_bound}) we bound the integrands of Expressions \eqref{eq:vsmall1}, \eqref{eq:vsmall2}, \eqref{eq:vsmall3}, and \eqref{eq:vsmall4}. For some $F$ we have
\begin{align*}
&|(\omega^T H^T \Sigma_K H \omega)^2| \leq F||\omega||_{\infty}^4||H||_{\infty}^4,\\
&|R(H\omega)(\omega^T H^T \Sigma_K H \omega)| \leq F||\omega||_{\infty}^6||H||_{\infty}^6,\\
&|R(H\omega)| \leq F||\omega||_{\infty}^4||H||_{\infty}^4,\\
&|R^2(H\omega)| \leq F||\omega||_{\infty}^8||H||_{\infty}^8.
\end{align*}
Note that by the definition of $d\nu(\omega)$ and the fact $\int ||\omega||_{\infty}^8 |f_{\epsilon}(\omega)|^2d\omega < \infty$ (Assumption \eqref{assump:1ep}) we have
\begin{equation*}
\int ||\omega||_{\infty}^8 d\nu(\omega) = \int ||\omega||_{\infty}^8 |\widehat{f}_{\epsilon}(\omega)|^2 d\omega - \int ||\omega||_{\infty}^8 |\widehat{f}_{\epsilon}(\omega)|^2 |\widehat{f}_X(\omega)|^2  d\omega < \infty.
\end{equation*}
So Expressions \eqref{eq:vsmall1}, \eqref{eq:vsmall2}, \eqref{eq:vsmall3}, and \eqref{eq:vsmall4} are all integrable and $O(||H||_{\infty}^4)$. Thus
\begin{equation*}
\frac{1}{n} \int |\widehat{K}(H\omega)|^2d\nu(\omega) = \frac{1}{n}\int d\nu(\omega) - \left(\frac{1}{n}\int (\omega^T H^T \Sigma_K H \omega) d\nu(\omega)\right)(1 + O(||H||_{\infty}^2)).
\end{equation*}
\hfill $\qed$

\subsection{Proof of Theorem \ref{thm:normal_mise}}
\label{sec:normal_mise_proof}
Recall that $K = \phi_{\Sigma_K}$, $f_{\epsilon} = \phi_{\Sigma_{\epsilon}}$, and 
\begin{equation}
\label{eq:fx_norm_mix}
f_X(x) = \sum_{j=1}^m \w_j \phi_{\Sigma_j}(x - \mu_j).
\end{equation}
Let $\w = (\w_1, \ldots, \w_m)$, $S = H^T \Sigma_K H$, and $\Omega_a$ for $a \in \{0,1,2\}$ be a $m \times m$ matrix with $j,j'$ entry equal to 
\begin{equation}
\label{eq:def_omegaa}
\phi_{aS + 2\Sigma_{\epsilon} + \Sigma_{j} + \Sigma_{j'}}(\mu_j - \mu_{j'}).
\end{equation}
In \textbf{Variance} we show
\begin{equation}
\label{eq:var_normal}
\int \Var(\fy) dy = \frac{1}{n} \left(\phi_{2S + 2\Sigma_\epsilon}(0) - \w^T \Omega_2 \w\right).
\end{equation}
In \textbf{Bias} we show
\begin{equation}
\label{eq:bias_normal}
\int \left(\E[\fy] - f_{Y}(y)\right)^2dy = \w^T (\Omega_2 - 2\Omega_1 + \Omega_0)\w.
\end{equation}
By the bias--variance decomposition of the $\MISE$, we can sum Equation \eqref{eq:var_normal} and \eqref{eq:bias_normal} to obtain the result
\begin{equation*}
\MISE(H) = \frac{1}{n} \phi_{2S + 2\Sigma_\epsilon}(0) + \w^T ((1 - n^{-1}) \Omega_2 - 2\Omega_1 + \Omega_0)\w.
\end{equation*}
First note that since $\epsilon$ and $K_H$ are both normal, the estimator $\fyh$ has a simple form, specifically
\begin{equation}
\label{eq:simple_est}
\fy = \frac{1}{n}\sum_{i=1}^n\phi_{S+\Sigma_{\epsilon}}(y - X_i).
\end{equation}
Second note that (see e.g., A.2 on \page 527 in \cite{wand1993comparison}) for any two covariance matrices $\Sigma$ and $\Sigma'$ and mean vectors $\mu$ and $\mu'$ we have
\begin{equation}
\label{eq:norm_int_id}
\int \phi_\Sigma(x - \mu)\phi_{\Sigma'}(x - \mu')dx = \phi_{\Sigma + \Sigma'}(\mu - \mu').
\end{equation}

\vspace{.1in}
\noindent
\textbf{Variance:}
Using Equation \eqref{eq:simple_est}, we have
\begin{equation}
\label{eq:var1234}
\int \Var(\fy) dy = \frac{1}{n}\left(\int \E[\phi_{S + \Sigma_\epsilon}^2(y - X_1)] dy - \int \E[\phi_{S + \Sigma_\epsilon}(y - X_1)]^2 dy \right).
\end{equation}
We now simplify each term in the parenthesis on the right hand side of this equation. Using Equation \eqref{eq:norm_int_id} for the last equality, for the first term on the right hand side of Equation \eqref{eq:var1234} we have
\begin{align*}
\int \E[\phi_{S + \Sigma_\epsilon}^2(y - X_1)] dy &=  \int \int \phi_{S + \Sigma_\epsilon}^2(y - x) f_X(x)dx dy  \\
&= \int \phi_{S + \Sigma_\epsilon}^2(y) dy\\
&= \phi_{2S + 2\Sigma_\epsilon}(0).
\end{align*}
Recalling the representation of $f_X$ in Equation \eqref{eq:fx_norm_mix}, the identity in Equation \eqref{eq:norm_int_id}, and the definition of $\Omega_2$ in Equation \eqref{eq:def_omegaa}, for the second term on the right hand side of Equation \eqref{eq:var1234} we have
\begin{align*}
\int \E[\phi_{S + \Sigma_\epsilon}(y - X_1)]^2dy &=\int \left( \int \phi_{S + \Sigma_{\epsilon}}(y - x) f_X(x) dx\right)^2dy\\
&=\int \left( \int \phi_{S + \Sigma_{\epsilon}}(y - x) \left( \sum_{j=1}^m \w_j \phi_{\Sigma_j}(x - \mu_j) \right) dx\right)^2dy\\
&= \int \left( \sum_{j=1}^m \w_j \phi_{S + \Sigma_{\epsilon} + \Sigma_j}(y - \mu_j)\right)^2 dy\\
&= \sum_{j=1}^m \sum_{j'=1}^m  \w_j \w_{j'}  \int\phi_{S + \Sigma_{\epsilon} + \Sigma_j}(y - \mu_j)\phi_{S + \Sigma_{\epsilon} + \Sigma_{j'}}(y - \mu_{j'}) dy\\
&= \sum_{j=1}^m \sum_{j'=1}^m  \w_j \w_{j'}  \phi_{2S + 2\Sigma_{\epsilon} + \Sigma_j + \Sigma_{j'}}(\mu_j - \mu_{j'})\\
&=\w^T \Omega_2 \w.
\end{align*}
Hence
\begin{equation*}
\int \Var(\fy) dy = \frac{1}{n} \left(\phi_{2S + 2\Sigma_\epsilon}(0) - \w^T \Omega_2 \w\right)
\end{equation*}

\vspace{.1in}
\noindent
\textbf{Bias:}
Recalling $f_{\epsilon} = \phi_{\Sigma_{\epsilon}}$ and the representations of $f_X$ and $\fyh$ in Equations \eqref{eq:fx_norm_mix} and \eqref{eq:simple_est}, note
\begin{align*}
&f_Y(y) = \int f_X(y - \epsilon)f_{\epsilon}(\epsilon)d\epsilon = \sum_{j=1}^m \w_j \int \phi_{\Sigma_j}(y - \epsilon - \mu_j)\phi_{\Sigma_\epsilon}(\epsilon)d\epsilon = \sum_{j=1}^m \w_j \phi_{\Sigma_j + \Sigma_\epsilon}(y - \mu_j),\\
&\E[\fy] = \int \phi_{S + \Sigma_{\epsilon}}(y - x) \sum_{j=1}^m \w_j \phi_{\Sigma_j}(x - \mu_j) dx = \sum_{j=1}^m \w_j \phi_{S + \Sigma_{j} + \Sigma_{\epsilon}}(y - \mu_j).
\end{align*}
Using these identities and the definition of $\Omega_a$ (Equation \eqref{eq:def_omegaa}), for the integrated squared bias we have
\begin{align*}
\int \left(\E[\fy] - f_{Y}(y)\right)^2dy &= \int \left( \sum_{j=1}^m \w_j \left( \phi_{S + \Sigma_j + \Sigma_{\epsilon}}(y-\mu_j) -  \phi_{\Sigma_j + \Sigma_{\epsilon}}(y-\mu_j)\right)\right)^2dy\\
&= \int \sum_{j=1}^m \sum_{j'=1}^m \w_j \w_{j'}\Big( \phi_{S + \Sigma_j + \Sigma_{\epsilon}}(y-\mu_j)\phi_{S + \Sigma_{j'} + \Sigma_{\epsilon}}(y-\mu_{j'}) \\
&- \phi_{S + \Sigma_j + \Sigma_{\epsilon}}(y-\mu_j)\phi_{\Sigma_{j'} + \Sigma_{\epsilon}}(y-\mu_{j'})\\
&- \phi_{S + \Sigma_{j'} + \Sigma_{\epsilon}}(y-\mu_{j'})\phi_{\Sigma_{j} + \Sigma_{\epsilon}}(y-\mu_{j})\\
&+\phi_{\Sigma_j + \Sigma_{\epsilon}}(y-\mu_j)\phi_{\Sigma_{j'} + \Sigma_{\epsilon}}(y-\mu_{j'})\Big)dy\\
&= \sum_{j=1}^m \sum_{j'=1}^m \w_j\w_{j'}\Big(\phi_{2S+2\Sigma_{\epsilon} + \Sigma_j + \Sigma_{j'}}(\mu_j - \mu_{j'})\\
& - 2\phi_{S + 2\Sigma_{\epsilon} + \Sigma_j + \Sigma_{j'}}(\mu_j - \mu_{j'}) + \phi_{2\Sigma_{\epsilon} + \Sigma_j + \Sigma_{j'}}(\mu_j - \mu_{j'})\Big)\\
&= \w^T (\Omega_2 - 2\Omega_1 + \Omega_0)\w.
\end{align*}

\subsection{Lemmas}
\begin{Lem}
\label{lemma:taylor_error_bound}
Under Assumptions \ref{assump:kernel}, $K$ is a symmetric density function in $\mathbb{R}^p$ with a characteristic function $\widehat{K}$ that is four times continuously differentiable. Let $\Sigma_K$ be the variance of $K$. We Taylor expand $\widehat{K}$ around $0$, obtaining
\begin{equation*}
  \widehat{K}(\omega) = 1 - \frac{\omega^T\Sigma_K \omega}{2} + R(\omega).
\end{equation*}
There exists $C$ such that for any $\omega$
\begin{equation*}
R(\omega) \leq C ||\omega||_{\infty}^4.
\end{equation*}
\end{Lem}
\begin{proof}
We bound the remainder term $R(\omega)$ by considering two cases.
\begin{enumerate}
\item $\{\omega:||\omega||_{\infty} \leq 1\}$: Since $\widehat{K}$ is four times continuously differentiable, there exists $D$ such that for any $\{j:\sum_{k=1}^p j_k=4\}$, $\forall \, ||\omega||_{\infty} \leq 1$
\begin{equation}
\label{eq:partial_bound}
\frac{\partial^4 \widehat{K}}{\partial \omega_1^{j_1},\ldots,\partial \omega_p^{j_p}} (\omega) < D.
\end{equation}
Using the mean value form of the Taylor remainder we have (see e.g. Theorem 7.1 in \cite{edwards1973advanced} on page 131)
\begin{equation*}
R(\omega) = \sum_{\{j:\sum_{k=1}^p j_k = 4\}} \frac{\partial^4 \widehat{K}}{\partial \omega_1^{j_1},\ldots, \partial \omega_p^{j_p}} (\xi) \prod_{k=1}^p \frac{\omega_k^{j_k}}{j_k!}.
\end{equation*}
for some $\xi = t \omega$ for $t \in [0,1]$.  Using Equation \eqref{eq:partial_bound} and noting $\prod_{k=1}^p \omega_k^{j_k} \leq ||\omega||_{\infty}^4$, for some $C$ we have
\begin{equation*}
|R(\omega)| \leq C||\omega||_{\infty}^4.
\end{equation*}
\item $\{\omega:||\omega||_{\infty} > 1\}$: Note that for some $D$, $\frac{\omega^T\Sigma_K \omega}{2} \leq D||\omega||_{\infty}^2$. Also note that on the set $||\omega||_{\infty} > 1$ we have $||\omega||_{\infty}^2 \leq ||\omega||_{\infty}^4$. We have
\begin{align*}
|R(\omega)| &= \left|\widehat{K}(\omega) - 1 + \frac{\omega^T\Sigma_K \omega}{2}\right|\\
&\leq |\widehat{K}(\omega)| + |1| + \left|\frac{\omega^T\Sigma_K \omega}{2}\right|\\
&\leq 2 + |\frac{\omega^T\Sigma_K \omega}{2}|\\
&\leq 2 + D||\omega||_{\infty}^2\\
&\leq 2||\omega||_{\infty}^2 + D||\omega||_{\infty}^2\\
&\leq (2 + D)||\omega||_{\infty}^4
\end{align*}
\end{enumerate}
\end{proof}

\section{Technical Notes}
\label{sec:technical_notes}

\subsection{Full Bandwidth Matrix Optimization}
\label{sec:full_bandwidth}
In Theorem \ref{thm:multi_dim} on \page \pageref{thm:multi_dim}, the $\MISE$ (using a full bandwidth matrix) is
\begin{equation*}
\frac{1}{n}\int d\nu(\omega) + \left(\frac{1}{4} \int (\omega^TS\omega)^2d\mu(\omega) - \frac{1}{n}\int (\omega^T S \omega) d\nu(\omega)\right)(1 + O(||H||_{\infty}^2))
\end{equation*}
where $S = H^T \Sigma_K H$. Using $\vc$ notation and the identity $\vc(EFG) = (G^T \otimes E) \vc(F)$ where $\otimes$ denotes Kronecker product (see Equation 5 on page 67 in \cite{henderson1979vec}), we write the optimization problem for $S$ as
\begin{equation}
\label{eq:full_opt}
S_Y^* = \argmin{S \succeq 0}  \vc(S)^T B \vc(S)   -\frac{1}{n} \vc(S)^T V\\
\end{equation}
where
\begin{align*}
&B = \frac{1}{4}\int (\omega \otimes \omega)(\omega \otimes \omega)^Td\mu(\omega),\\
&V = \int (\omega \otimes \omega) d\nu(\omega).
\end{align*}
It is important to note that $B$ and $V$ cannot be computed from the data because they depend on the unknown function $\widehat{f}_X(\omega)$. In practice we could use plug--in estimators to approximate these integrals.

The unconstrained solution to optimization problem \eqref{eq:full_opt} may not be positive semidefinite, so we cannot omit the $S \succeq 0$ constraint and use a quadratic solver (see Subsection \ref{sec:diagonal} for an example). Also note that one cannot analytically solve the unconstrained version of optimization problem \eqref{eq:full_opt} and then check whether the resulting $S_Y^*$ is positive semidefinite. In other words, the following procedure is not valid:
\begin{align*}
&g(\vc(S)) \equiv \vc(S)^T B \vc(S) -\frac{1}{n} \vc(S)^T V,\\
\implies&\nabla g(\vc(S)) = 2B\vc(S) -\frac{1}{n}  V.
\end{align*}
Setting the gradient equal to $0$ and solving we have
\begin{equation*}
\vc(S_Y^*) =  \frac{1}{2n}B^{-1}V.
\end{equation*}
One could then check whether $S_Y^* \succeq 0$. This procedure is not valid because $B$ is not invertible. To see that $B$ is not invertible, note that the vector $(\omega \otimes \omega)$ has $p^2$ elements, but not $p^2$ unique elements. For example when $p=2$, $(\omega \otimes \omega) = (\omega_1^4,\omega_1\omega_2,\omega_1\omega_2,\omega_2^2)^T$. When the $j$th and $k$th elements of $(\omega \otimes \omega)$ are equal, the $j$th and $k$th rows of $(\omega \otimes \omega) (\omega \otimes \omega)^T$ are equal. Thus at least two rows of $B\equiv\int (\omega \otimes \omega)(\omega \otimes \omega)^Td\mu(\omega)$ are equal, implying that $B$ cannot be inverted.

\subsection{Diagonal Bandwidth and $\Sigma_K = I$}
\label{sec:diagonal}

In Equation \eqref{eq:full_bandwidth_mise} the full bandwidth matrix asymptotic expansion of the $\MISE$ was presented. By restricting the kernel to have $\Sigma_K = I$ and the bandwidth matrix to be diagonal we achieve considerable simplification of the $\MISE$. Let $h_i = H_{ii}$ and $\hsq = (h_1^2, \ldots, h_p^2)$. The $\MISE$ (Equation \eqref{eq:full_bandwidth_mise}) becomes
\begin{equation*}
(2\pi)^p \MISE(\hsq) = \frac{1}{n}\int d\nu(\omega) + \left(\hsq^TB\hsq - \frac{1}{n}\hsq^T V\right)(1 + O(||\hsq||_{\infty})),
\end{equation*}
where
\noindent
\begin{align*}
&B_{i,j} = \frac{1}{4}\int \omega_i^2\omega_j^2d\mu(\omega),\\
&V_i = \int \omega_i^2 d\nu(\omega).
\end{align*}
We seek the $\hsq$ which minimizes the larger order terms in the $\MISE$ expression. In other words we seek
\begin{equation}
\label{eq:diag_optimal}
\hsq^*_Y = \argmin{\hsq \geq 0} \left( \hsq^TB\hsq - \frac{1}{n}\hsq^TV\right).
\end{equation}
$B$ is positive definite so the expression is strictly convex and there is a unique solution. Enforcing the domain restriction $\hsq \geq 0$ (elementwise) is necessary: even in simple cases, the unconstrained optimum $\frac{1}{2n}B^{-1}V$ may have elements less than $0$. In the following paragraphs we work through an example where $f_X$ and $f_{\epsilon}$ are bivariate independent normals with $\epsilon$ having small variance along one direction. The kernel is normal with identity covariance. The normality is not essential for this example, but makes the computations simpler.


We begin by showing that the optimal bandwidth matrix is diagonal, implying that optimizing over the full bandwidth matrix and the diagonal matrix are equivalent. We then show that when optimizing over the unconstrained diagonal matrix, the direction in which $\epsilon$ has larger variance yields a ``negative squared bandwidth''. Consider:
\begin{align*}
&f_X \sim N(0,I_{2 \times 2}),\\
&f_{\epsilon} \sim N(0,
\begin{bmatrix}
\sigma_1^2 & 0 \\
0 & \sigma_2^2 \\
\end{bmatrix}),\\
&\Sigma_K \equiv \int xx^TK(x)dx = I{2 \times 2}.
\end{align*}
We parameterize the bandwidth matrix using $ H = 
\left[\begin{array}{cc}
h_{11} & h_{12}\\
h_{12} & h_{22}
\end{array}\right]
$. First consider optimizing over the entire bandwidth matrix, Equation \eqref{eq:full_opt}.  In our case
\begin{align*}
&S \equiv H^T \Sigma_K H = H^TH,\\
&B = \int (\omega \otimes \omega)(\omega \otimes \omega)^Td\mu(\omega) = \int \begin{bmatrix}\omega_1^4 & \omega_1^3\omega_2 & \omega_1^3\omega_2 & \omega_1^2\omega_2^2 \\
\omega_1^3\omega_2 & \omega_1^2\omega_2^2 & \omega_1^2\omega_2^2 &\omega_1\omega_2^3 \\
\omega_1^3\omega_2 & \omega_1^2\omega_2^2 & \omega_1^2\omega_2^2 &\omega_1\omega_2^3 \\
\omega_1^2\omega_2^2 & \omega_1\omega_2^3 & \omega_1\omega_2^3 &\omega_2^4
\end{bmatrix}d\mu(\omega),\\
&V =\int (\omega \otimes \omega) d\nu(\omega) = \int \begin{bmatrix}
\omega_1^2 \\
\omega_1\omega_2 \\
\omega_1\omega_2 \\
\omega_2^2
\end{bmatrix} d\nu(\omega).
\end{align*}
So Equation \eqref{eq:full_opt} becomes
\begin{align*}
&\vc(H^TH)^T\int\begin{bmatrix}\omega_1^4 & \omega_1^3\omega_2 & \omega_1^3\omega_2 & \omega_1^2\omega_2^2 \\
\omega_1^3\omega_2 & \omega_1^2\omega_2^2 & \omega_1^2\omega_2^2 &\omega_1\omega_2^3 \\
\omega_1^3\omega_2 & \omega_1^2\omega_2^2 & \omega_1^2\omega_2^2 &\omega_1\omega_2^3 \\
\omega_1^2\omega_2^2 & \omega_1\omega_2^3 & \omega_1\omega_2^3 &\omega_2^4
\end{bmatrix}d\mu(\omega) \vc(H^TH) \\
-&\frac{1}{n}  \vc(H^TH)^T \left(\int \begin{bmatrix}
\omega_1^2 \\
\omega_1\omega_2 \\
\omega_1\omega_2 \\
\omega_2^2
\end{bmatrix} d\nu(\omega)\right).
\end{align*}
The integration causes those terms involving odd powers of $\omega_i$ to be 0 by independence and symmetry of $d\nu(\omega)$ and $d\mu(\omega)$. Additionally the center $\omega_1^2\omega_2^2$ terms are moved outside the main expression and into the third term. We have
\begin{align*}
& \vc(H^TH)^T\int\begin{bmatrix}\omega_1^4 & 0 & 0 & \omega_1^2\omega_2^2 \\
0 & 0 & 0 & 0 \\
0 & 0 & 0 & 0 \\
\omega_1^2\omega_2^2 & 0 & 0 &\omega_2^4
\end{bmatrix}d\mu(\omega) \vc(H^TH) \\
&- \frac{1}{n} \vc(H^TH)^T\left(\int \begin{bmatrix}
\omega_1^2 \\
0 \\
0 \\
\omega_2^2
\end{bmatrix} d\nu(\omega)\right)
+ 4 (h_{12}(h_{11} + h_{22}))^2\int \omega_1^2 \omega_2^2 d\mu(\omega).
\end{align*}
Since 
\begin{equation*}
H^TH = \left[\begin{array}{cc}
h_{11}^2 + h_{12}^2 & h_{12}(h_{11} + h_{22})\\
h_{12}(h_{11} + h_{22}) & h_{12}^2 + h_{22}^2
\end{array}\right],
\end{equation*}
minimization of the first two terms depends on $(h_{11}^2 + h_{12}^2,h_{22}^2 + h_{12}^2)$. So by setting $h_{12}=0$ we make the third term in the expression 0, without restricting minimization of the first two terms. Thus for the general bandwidth matrix the minimum occurs when the off-diagonal elements are 0.

Now let $\hsq =(h_{11}^2,h_{22}^2)$. We study the diagonal optimization problem \eqref{eq:diag_optimal}
\begin{equation*}
\hsq_Y^* = \min_{\hsq} \hsq^TB'\hsq - \frac{1}{n}\hsq^TV',
\end{equation*}
where
\noindent
\begin{align*}
&B'_{i,j} = \frac{1}{4}\int \omega_i^2\omega_j^2d\mu(\omega) = \frac{1}{4}\int \omega_i^2\omega_j^2|\widehat{f}_X(\omega)|^2 |\widehat{f}_{\epsilon}(\omega)|^2d\omega,\\
&V'_i = \int \omega_i^2 d\nu(\omega) = \int \omega_i^2|\widehat{f}_{\epsilon}(\omega)|^2 d\omega  - \int \omega_i^2 |\widehat{f}_{X}(\omega)|^2 |\widehat{f}_{\epsilon}(\omega)|^2 d\omega.
\end{align*}
With no restrictions on $\hsq$ the optimum is
\begin{equation*}
\hsq_Y^* = \frac{1}{2n}B'^{-1}V'.
\end{equation*}
We now compute this quantity for the given densities. First compute $B'$:
\begin{align*}
4B'_{11} &= \int \omega_1^4|\widehat{f}_{X_1}(\omega_1)|^2 |\widehat{f}_{\epsilon_1}(\omega_1)|^2d\omega_1  \int |\widehat{f}_{X_2}(\omega_2)|^2 |\widehat{f}_{\epsilon_2}(\omega_2)|^2d\omega_2 \\
&= \left(\frac{3}{4} \sqrt{\frac{\pi}{(1 + \sigma_1^2)^5}}\right)\left(\sqrt{\frac{\pi}{1 + \sigma_2^2}}\right),\\
4B'_{22} &= \int \omega_2^4|\widehat{f}_{X_2}(\omega_2)|^2 |\widehat{f}_{\epsilon_2}(\omega_2)|^2d\omega_2  \int |\widehat{f}_{X_1}(\omega_1)|^2 |\widehat{f}_{\epsilon_1}(\omega_1)|^2d\omega_1 \\
&= \left(\frac{3}{4} \sqrt{\frac{\pi}{(1 + \sigma_2^2)^5}}\right)\left(\sqrt{\frac{\pi}{1 + \sigma_1^2}}\right),\\
4B'_{12} &= \int \omega_1^2|\widehat{f}_{X_1}(\omega_1)|^2 |\widehat{f}_{\epsilon_1}(\omega_1)|^2d\omega_1  \int \omega_2^2|\widehat{f}_{X_2}(\omega_2)|^2 |\widehat{f}_{\epsilon_2}(\omega_2)|^2d\omega_2\\ 
&= \left(\frac{1}{2} \sqrt{\frac{\pi}{(1 + \sigma_1^2)^3}}\right) \left(\frac{1}{2} \sqrt{\frac{\pi}{(1 + \sigma_2^2)^3}}\right).
\end{align*}
Since $B'$ and $B'^{-1}$ are symmetric, we write only the upper triangle:
\begin{equation*}
B' = \frac{\pi}{16}\begin{bmatrix}
3 \frac{1}{\sqrt{(1 + \sigma_1^2)^5(1 + \sigma_2^2)}} & \frac{1}{\sqrt{(1 + \sigma_1^2)^3(1 + \sigma_2^2)^3}} \\
 & 3 \frac{1}{\sqrt{(1 + \sigma_2^2)^5(1 + \sigma_1^2)}} \\
\end{bmatrix}.
\end{equation*}
Taking the inverse we obtain
\begin{align*}
B'^{-1} &= \frac{2(1+\sigma_1^2)^3(1+\sigma_2^2)^3}{\pi}\begin{bmatrix}
3 \frac{1}{\sqrt{(1 + \sigma_2^2)^5(1 + \sigma_1^2)}} & -\frac{1}{\sqrt{(1 + \sigma_1^2)^3(1 + \sigma_2^2)^3}} \\
 & 3 \frac{1}{\sqrt{(1 + \sigma_1^2)^5(1 + \sigma_2^2)}} \\
\end{bmatrix}\\
&=\frac{2}{\pi}\begin{bmatrix}
3 \sqrt{(1 + \sigma_2^2)(1 + \sigma_1^2)^5} & -\sqrt{(1 + \sigma_1^2)^3(1 + \sigma_2^2)^3} \\
 & 3 \sqrt{(1 + \sigma_1^2)(1 + \sigma_2^2)^5} \\
\end{bmatrix}.
\end{align*}
For $V'$ we have
\begin{align*}
V' &= \frac{\pi}{2}\left(\begin{bmatrix}
\sigma_1^{-3}\sigma_2^{-1} \\
\sigma_1^{-1}\sigma_2^{-3} \\
\end{bmatrix} - 
\begin{bmatrix}
\frac{1}{\sqrt{(1 + \sigma_1^2)^{3}(1 + \sigma_2^2)}} \\
\frac{1}{\sqrt{(1 + \sigma_2^2)^{3}(1 + \sigma_1^2)}} \\
\end{bmatrix}\right)\\
&= \frac{\pi}{2\sigma_2^3}\left(\begin{bmatrix}
0\\
\sigma_1^{-1} \\
\end{bmatrix} + \sigma_2^{2}
\begin{bmatrix} 
\sigma_1^{-3} \\
0 \\
\end{bmatrix} - \sigma_2^3
\begin{bmatrix}
\frac{1}{\sqrt{(1 + \sigma_1^2)^{3}(1 + \sigma_2^2)}} \\
\frac{1}{\sqrt{(1 + \sigma_2^2)^{3}(1 + \sigma_1^2)}} \\
\end{bmatrix}\right).
\end{align*}
So the optimal $\hsq$ is
\begin{align*}
\hsq_Y^* &= \frac{1}{2n}B'^{-1}V' \\
&= \frac{1}{2n\sigma_2^3}\bigg(
\begin{bmatrix}
-\sigma_1^{-1} \sqrt{(1 + \sigma_1^2)^3(1 + \sigma_2^2)^3}\\
3 \sigma_1^{-1} \sqrt{(1 + \sigma_1^2)(1 + \sigma_2^2)^5}
\end{bmatrix}\\
&+ \sigma_2^2
\begin{bmatrix}
3 \sigma_1^{-3} \sqrt{(1 + \sigma_2^2)(1 + \sigma_1^2)^5}\\
-\sigma_1^{-3} \sqrt{(1 + \sigma_1^2)^3(1 + \sigma_2^2)^3}
\end{bmatrix}
- 2\sigma_2^3
\begin{bmatrix}
1 + \sigma_1^2\\
1 + \sigma_2^2
\end{bmatrix}\bigg).
\end{align*}
For $\sigma_2$ close to 0 and small relative to $\sigma_1$ this quantity is approximately
\begin{equation}
\label{eq:hoptapprox}
\hsq_Y^* \approx \frac{1}{2n\sigma_1\sigma_2^3}\left(
\begin{bmatrix}
-\sqrt{(1 + \sigma_1^2)^3}\\
3 \sqrt{(1 + \sigma_1^2)}
\end{bmatrix}\right).
\end{equation}
The unconstrained optimization results in an $\hsq_Y^*$ with negative elements.

\end{document}